\documentclass[a4]{article}
\usepackage{a4wide}
\usepackage[utf8]{inputenc}
\usepackage{amsthm}
\usepackage{cleveref}
\crefname{lemma}{lemma}{lemmas}
\Crefname{lemma}{Lemma}{Lemmas}
\crefname{corollary}{corollary}{Corollaries}
\Crefname{corollary}{corollary}{Corollaries}
\Crefname{algorithm}{Algorithm}{Algorithms}
  
\usepackage{authblk}
\usepackage{color}
\usepackage{graphicx}
\usepackage{url}%rely
\usepackage{amsmath,amssymb,latexsym}
\usepackage[vlined,english,boxed]{algorithm2e}
\usepackage{amsmath} % for styled algorithms
      \IncMargin{1em}  % margin of algorithms
      \SetKwInput{KwData}{Input}
      \SetKwInput{KwResult}{Output}
      \SetKwIF{If}{ElseIf}{Else}{If}{then}{else if}{else}{endif}
      \SetKwFor{ForEach}{For each}{do}{done}
      \SetKwFor{While}{While}{do}{done}
      \SetKwFor{Procedure}{Procedure}{}{}
      \SetKw{Return}{Return}%
      \SetKwComment{Comment}{/* }{ */}
      \SetSideCommentRight
      \DontPrintSemicolon

\usepackage{tikz}

\def\comment#1#2{{\color{red}[#1: {#2}]}}

\def\leonardo#1{\comment{LL}{#1}}
\def\leonardo#1{}
\long\def\jump#1\finjump{}
\theoremstyle{definition}
\newtheorem{claim}{\emph{Claim}}
\newtheorem{proposition}{Proposition}

\newtheorem{theorem}{Theorem}
\newtheorem{lemma}{Lemma}
\newtheorem{corollary}{Corollary}
\newcommand{\card}[1]{\left|{#1}\right|}

\newcommand{\paren}[1]{\left({#1}\right)}

\newcommand{\set}[1]{\left\{{#1}\right\}}

\newcommand{\A}{\mathcal{A}}

\newcommand{\D}{\mathcal{D}}
\newcommand{\R}{\mathcal{R}}
\newcommand{\C}{\mathcal{C}}
\newcommand{\UC}{\mathcal{UC}}

\newcommand{\I}{\mathcal{I}}

\newcommand{\mb}[1]{\mbox{\it #1}}
\newcommand{\average}[1]{\overline{#1}}

\renewcommand{\k}{\average{k}}
\newcommand{\K}{\average{K}}

\title{Efficient Loop Detection in Forwarding Networks and Representing Atoms in a Field of Sets}

\author{Yacine Boufkhad\thanks{Université Paris Diderot}
  \and Leonardo Linguaglossa\thanks{Inria}
  \and Fabien Mathieu\thanks{Nokia Bell Labs}
  \and Diego Perino\thanks{Telefonica Research}
  \and Laurent Viennot\thanks{Inria}}

\begin{document}

\maketitle

\begin{abstract}
The problem of detecting loops in a forwarding network is known to be
NP-complete when general rules such as wildcard expressions are used. Yet,
network analyzer tools such as Netplumber (Kazemian et al., NSDI'13) or
Veriflow (Khurshid et al., NSDI'13) efficiently solve this problem in networks
with thousands of forwarding rules.  In this paper, we complement such
experimental validation of practical heuristics with the first provably
efficient algorithm in the context of general rules.  Our main tool is a
canonical representation of the atoms (i.e. the minimal non-empty sets) of the
field of sets generated by a collection of sets. This tool is particularly
suited when the intersection of two sets can be efficiently computed and
represented. 
%% This is the case for forwarding networks, where the predicate
%% filter associated to a rule can be seen seen as a compact data-structure
%% representing the set of headers that match the rule.  The header classes of the
%% network (sets of headers matching the same rules) embrace all possible
%% forwarding behaviors and their number measures how many tests are classically
%% performed for forwarding loop detection. These classes are indeed the atoms of
%% the field of sets generated by the collection of all predicate filters of the
%% network.  Following this equivalence, we first provide an efficient
%% representation of atoms that allows to obtain the 
In the case of forwarding networks, each forwarding rule is associated with
the set of packet headers it matches. The atoms then correspond to classes
of headers with same behavior in the network. We propose an algorithm for
atom computation and provide the first polynomial time
algorithm for loop detection in terms of number of classes (which can
be exponential in general). This contrasts
with previous methods that can be exponential, even in simple cases with linear
number of classes.  Second, we introduce a notion of network dimension captured
by the overlapping degree of forwarding rules. The values of this measure
appear to be very low in practice and constant overlapping degree ensures
polynomial number of header classes. Forwarding loop detection is thus
polynomial in forwarding networks with constant overlapping degree.
\end{abstract}


\section{Introduction} % -----------------------------------------------------

With the multiplication of network protocols, %that interact within the same
%devices,
 network analysis has become an important and challenging task.  We
focus on a key diagnosis task: detecting possible forwarding loops. Given a
network and node forwarding tables, the problem
consists in testing whether there exists a packet header $h$ and a directed
cycle in the network topology such that a packet with header $h$ will
indefinitely loop along the cycle.
%Assuming that forwarding tables are ordered
%according to priority of rules, this means that the first rule matched by $h$
%in the table of each node in the cycle indicates to forward packets with header
%$h$ to the next node in the cycle.
This problem is indeed NP-complete as noted
in \cite{anteater}.  Its hardness comes from the use of compact representations
for predicate filters: the set of headers that match a rule is classically
represented by a prefix in IP forwarding, a general wildcard expression in
Software-Defined Networking (SDN), value ranges in firewall rules, or even
a mix of such representations if several header fields are considered. 
%% Routing in
%% practical networks indeed takes into account several fields and may mix this
%% various representations. In this paper, we will focus on wildcard
%% expressions for the clarity of the discussion as wildcard expressions for
%% separate header fields can be concatenated to a single one. However, the
%% results generalize to most classical compact data structures for representing
%% sets of bit strings.

We first give a toy example of forwarding loop problem where
the predicate filter of each rule is given by a wildcard expression, that is an $\ell$-letter string in $\set{1,0,*}^\ell$. Such an expression represents
the set all $\ell$-bit headers obtained by replacing each $*$ of the
expression by either $0$ or $1$. A packet with header in that set is said to
match the rule.
%It is associated with the action to be taken
%on packets with header in that set (such packets or headers are 
%said to match the rule): drop, forward to a neighbor or deliver
%locally.
%as drop or forward that should be taken on packets 
% (Wildcard matching can be seen as generalization of prefix matching.)
Figure~\ref{fig:one_node_loop} illustrates a one node network with wildcard
expressions of $\ell=4$ letters. Rules are tested from top to bottom.  All rules
indicate to drop packets except the last one %that applies to any header and
%indicates to 
that forwards packets to the node itself. This network contains a
forwarding loop if there exists a header $x_1x_2x_3x_4$ that %the last rule can applied to some header $x_1x_2x_3x_4$
matches no rule except the last one.
%More precisely,
%this means that $x_1x_2x_3x_4$ must match the last rule and none of the
%previous ones. 
Not matching a rule as $110*$ corresponds to having
$x_1=0$, $x_2=0$, or $x_3=1$.
%, or equivalently to satisfy that the boolean formula
%$\overline{x_1}\vee\overline{x_2}\vee x_3$.  
This one node network
 thus has a forwarding loop iff the formula
$(\overline{x_1}\vee\overline{x_2}\vee\overline{x_3}\vee\overline{x_4})\wedge
(\overline{x_1}\vee\overline{x_2}\vee\overline{x_3}\vee x_4)\wedge
(\overline{x_1}\vee\overline{x_2}\vee x_3)\wedge (\overline{x_1}\vee x_2)\wedge
(x_1)$ is satisfiable, which is not the case.
% and the network has no
% forwarding loop even though its topology has a directed cycle. 
This simple example can easily be generalized to reduce SAT to forwarding loop
detection in networks with wildcard rules. It also points out a key
problem: testing the emptiness of expressions such as
$r_p\setminus\cup_{i=1..p-1}r_i$ where $r_1,\ldots,r_p$ are the sets associated
to $p$ rules.

\vspace{-5mm}
\begin{figure}[h]
  \begin{center}
    \includegraphics[width=.4\textwidth]{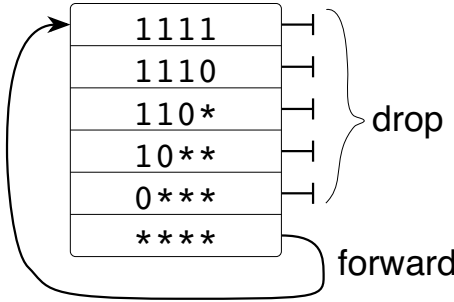}
    \caption{Does this one node network have a forwarding loop ?}
    \label{fig:one_node_loop}
  \end{center}
\end{figure}
\vspace{-5mm}

As packet headers in practical networks such as Internet typically have
hundreds of bits, % (see for example Table 12 from \cite{openflow}), exhaustive
search of the header space is completely out of reach.  The main challenge for
solving such a problem thus resides in limiting the number of tests to
perform. For that purpose, previous works~\cite{veriflow,hsa} propose to
consider sets of headers that match some predicate filters and do not match some
others. Defining two headers as equivalent when they match exactly the same
predicate filters, it then suffices to perform one test per equivalence
class. These classes are indeed the atoms (the minimal non-empty sets) of the
field of sets (the (finite) $\sigma$-algebra) generated by the sets associated
to the rules. 
%This can be easily seen by noting that the class of headers
%matching a given subset $R$ of rules is the intersection of the sets associated
%to rules in $R$
%intersected with the intersection of
%and the complements of the sets associated to rules not in $R$.

A first challenge lies in efficiently identifying and representing these
atoms. This would be fairly easy if both intersection and complement could be
represented efficiently.
In practice, most classical compact data-structures for sets of bit strings are closed under intersection but not under complement.
For example, the intersection of two wildcard expressions, if not empty,  can obviously be represented by a wildcard expression, but the complement
of a wildcard expression is more problematic.
Previous works overcome this difficulty by
representing the complement of a $\ell$-letter wildcard expression as the union of
several wildcard expressions (up to $\ell$).  
However, this can result in exponential blow-up
and the tractability of these methods rely on various heuristics that do not
offer rigorously proven guarantees.

A second challenge lies in understanding the tractability of practical
networks. One can easily design a collection of $2\ell$ wildcard
expressions that generates all the $2^\ell$ possible singletons as atoms
(all the $\ell$-letters strings with only one non-$*$ letter).
What does prevent such phenomenon in practice?
%However, 
Can we provide a property that intuitively fits with practical network and
guarantees  that the number of atoms does not blow up? 
%Can we efficiently compute the atoms of the $\sigma$-algebra generated
%by a collection of sets (and benefit from such a property)?
This paper aims at addressing both challenges with provable guarantees.

\paragraph*{Related work}
The interest for network problem diagnosis has recently grown after
the advent of Software-Defined Networking (SDN)~\cite{Peresini:2013:OCP:2491185.2491205,Kotronis:2012:ORC:2390231.2390241,Monsanto:2013:CSN:2482626.2482629,Katta:2013:ICU:2491185.2491191}. SDN offers the opportunity to manage the forwarding tables of a network with a centralized controller where full knowledge is available for analysis. 
%% It also raises new challenges
%% by standardizing the use of arbitrary wildcard expressions in forwarding rules.
%% Several tools have been proposed to test various mis-behaviors such
%% forwarding loops or black-holes (incorrect drop of packets), or to answer
%% specific tests such as reachability between two nodes. Similar techniques
%% are generally used for all tests and loop detection appears as an emblematic
%% and representative problem.
%
Previous works has led to a series of methods for 
network analysis, resulting in several 
tools~\cite{veriflow,netplumber,anteater,libra}.  
%Most of them mainly
The main approaches rely on computing classes of headers 
%that match some rules and not some others 
by combining rule predicate filters using intersection and set difference
(that is intersection with complement). 
%While these methods appear to be effective in practice, their performance
%cannot be guaranteed as the representation of a single class can explode
%exponentially when repeatedly using set difference operations.
%The main difference of our
%approach is to avoid complement or set difference operations and rely solely
%on intersection and cardinal operations to enable performance guarantees.
The idea of considering all header classes generated by the global collection
of the sets associated to all forwarding rules in the network is due to
Veriflow~\cite{veriflow}. However, the use of set differences results in
computing a refined partition of the atoms of the field of sets generated by
this collection that can be much larger than an exact representation. 
NetPlumber~\cite{netplumber}, which relies on the header space analysis
introduced in~\cite{hsa}, refines this approach by considering the set of
headers that can follow a given path of the network topology. This set is
represented as a union of classes that match some rules (those that indicate to
forward along the path) and not some others (those that have higher priority
and deviate from the path): a similar problem of atom representation thus
arises.  The idea of avoiding complement operations is somehow approached in
the optimization called ``lazy subtraction'' that consists in delaying as much
as possible the computation of set differences. However, when a loop is
detected, formal expressions with set differences have to be tested for
emptiness. They are then actually developed, possibly resulting in the
manipulation of expressions with exponentially many terms.
%This approach of
%refining classes according to the topology allows to manage rules with write
%actions that can modify the packet header.  

Concerning the tractability of the problem, the authors of
NetPlumber observe a phenomenon
called ``linear fragmentation''~\cite{hsa,hsa-report} that allows to argue for the
efficiency of the method. They introduce a parameter $c$ measuring this linear
fragmentation and claim a polynomial time bound for loop detection for low
$c$~\cite{hsa} (when emptiness tests are not included in the analysis).
However, the rigorous analysis provided in~\cite{hsa-report}
includes a $c^{D_G}$ factor where $D_G$ is the diameter of the network
graph. While this factor appears to be largely overestimated in practice,
the sole hypothesis of linear fragmentation does not suffice for
explaining tractability and prove polynomial time guarantees.
The alternative approach of
Veriflow is specifically optimized for rules resulting from range matching
within each field of the header. When the number of fields is constant,
polynomial time execution can be guaranteed but this result does not
extend to general wildcard matching.

%% Previous work
%% observe a phenomenon called ``linear fragmentation''
%% and argue that the number of header classes tends to be polynomial in
%% practice~\cite{hsa}. However, an exponential term shows up in the 
%% analysis~\cite{hsa-report} and this model does not allow to explain
%% rigorously how tractability can arise.

%% Anteater~\cite{anteater} represents verification tasks as boolean
%% satisfiability problems and solve them using a generic SAT solver.  The boolean
%% formulas generated by this approach equivalently represent combinations of
%% rules with intersection and set difference.  Libra~\cite{libra} assumes that
%% rules are based on prefix matching and is concerned by solving the verification
%% tasks in a high performance perspective using MapReduce framework. One can show
%% that the number of header classes is linear in the number of rules when they
%% all are prefixes. The context of prefix matching rules is thus significantly
%% simpler.

A similar
% but simpler problem of header space analysis 
problem consists
in conflict detection between rules and their 
resolution~\cite{filterConflict,packetFilter,boutier}.
%% These problems are easier as they can be solved by  the sole use of
%% intersection operations.
%% in filter conflict detection and resolution.
%% Filter conflict detection ask to
%% determine if two rules with highest priority might indicate different actions
%% for a given header. Filter conflict resolution consists in adding 
%% more specific rules (with higher priority) for masking all conflicts.
It has mainly been studied in the context of multi-range
rules~\cite{filterConflict,packetFilter}, which can benefit from computational
geometry algorithms. 
(A multi-range can be seen as a hyperrectangle
%box with faces orthogonal to axis' 
in a $d$-dimensional euclidean space where $d$
is the number of fields composing headers.)
%In this special setting, it is
%possible to benefit form classical computational geometry algorithms such as
%intersection of orthogonal objects~\cite{intersectionBoxes} for retrieving
%efficiently the set of rules that intersect a given rule.  
Another similar problem, determining efficiently the rule that applies
to a given packet, has been %extensively 
studied for 
multi-ranges~\cite{packetFilter,packetClassif,packetClassifOverview}.
%
%Arbitrary wildcards of length $\ell$ can be seen as multi-ranges in a
%$\ell$-dimensional space by substituting the ranges $[0,0],[1,1],[0,1]$ to the
%letters $0,1,*$ respectively. However, computational geometry techniques are
%not competitive for such high dimension. The intersection problem that consists
%in finding the wildcard expressions that intersect a given wildcard expression
%(according to their associated sets) is related to the old problem of partial 
%matching~\cite{partialMatchRivest}. 
In the case of wildcard matching, such problems are related to the old problem of partial matching~\cite{partialMatchRivest}.
% which is equivalent to that of finding if one of the
% rules of a forwarding table matches a given header. 
It is believed to suffer from the ``curse of
dimensionality''~\cite{highDimNN,patrascu11} and no method significantly faster
than exhaustive search is expected to be found with near linear memory
(although some tradeoffs are known for small number of $*$
letters~\cite{newPartialMatch}). 
However, efficient hardware-based implementations
exist based on Ternary Content Addressable Memory
(TCAMs)~\cite{Bosshart2013sdn} or Graphics Processing Unit (GPU)~\cite{sdngpu}.
%Efficient hardware implementations are obtained using Ternary Content Addressable Memory (TCAMs) or even more advanced techniques~\cite{Bosshart2013sdn}.
%From a more practical perspective, wild-card matching or programmable match-action processing in high-speed network devices is usually implemented in hardware~\cite{Bosshart2013sdn}. 
%High-speed software-based lookup algorithms can also be obtained by using Graphics Processing Unit (GPU) capabilities~\cite{sdngpu}. 
%They enable full matching programmability and allow to store a larger number of rules that their hardware counterpart, but they still do not match the  maximum throughput performance of hardware-based designs.  
%For near linear space, In the
%cell-probe model, fast query time can only be attained at the cost high memory 
%usage~\cite{patrascuUnify}.
%\lv{ajouter ref TCAMs pr solutions hardware}

Regarding the manipulation of set collections, recent work~\cite{MaryS16}
shows how to enumerate all sets obtained by closure under given set operations
with polynomial delay.  In particular, this allows to produce the field of sets
generated by the collection. However, this setting requires a set
representation which explicitly lists all elements and does not apply here.
Another issue comes from the fact that the field set can be exponentially
larger than its number of atoms.

\paragraph*{Our contributions}
First, we make a key algorithmic step by providing an efficient algorithm for
computing an exact representation of the atoms of the field of sets
generated by a collection of sets. The representation obtained is linear in the
number of atoms and allows to test efficiently if an atom is included in a
given set of the collection. The main idea is to represent an atom by 
the intersection of the sets that contain it. We avoid complement computations
by using cardinality computations for testing emptiness.
Our algorithm is generic and supports
any data-structure for representing sets of $\ell$-bit strings that supports
intersection and cardinality computation in bounded time $O(T_\ell)$ for some
value $T_\ell$. It runs in polynomial time with respect to $n,m$, the number
of sets and atoms respectively.
Beyond combinations of wildcard and range expressions, we believe that it
could be extended to support expressions on hashed values (for a fixed hash
function) or bloom filters by defining auxiliary cardinality
measures. 
%% This algorithm allows to solve forwarding loop detection in $O(\ell
%% nm^2+nn_Gm)$ time where $n$ denotes the number of rules, $m$ the number of
%% header classes, and $n_G$ the number of nodes in the network.  The problem is
%% thus polynomial with respect to $n,m$.  

\begin{table}
  \begin{center}
    \scalebox{0.84}{
      \bgroup\small
      \begin{tabular}{|l|c|c|c|l|}
        \hline
        Rule repr.
        %& Header cl.
        & Trivial
        & NetPlumber \cite{netplumber}  & Veriflow \cite{veriflow} 
        & This paper \\ %(Th.~\ref{th:main})\\
        \hline
        $T_\ell$-bounded %& $m$ 
          & $O(T_\ell n n_G 2^\ell)$ &
          -- & -- & $O( T_\ell n m^2  + n n_G m)$ \\
        '' ov. deg. $O(\log n)$ %& $m=O(n^{k})$ 
          & '' & -- & -- 
          %& $O((T_\ell n + n_G) m) = O((T_\ell n + n_G) n^{k})$
          %& $O( (T_\ell (n + k2^{k}\log m)  + k n_G) m) $
          & $T_\ell n^{O(1)} m + O(n_G m \log n)$\\
        '' ov. deg. $k$ %& $m=O(n^{k})$ 
          & '' & -- & -- 
          %& $O((T_\ell n + n_G) m) = O((T_\ell n + n_G) n^{k})$
          %& $O( (T_\ell (n + k2^{k}\log m)  + k n_G) m) $
          & $O(( T_\ell n  + T_\ell k^22^{k}\log n + k n_G) n^{k})$\\
        \hline
        $\ell$-wildcard %& $m\le 2^{\min(\ell,n)}$ 
          & $O(\ell n n_G 2^\ell)$ & $\Omega(\ell n_G 2^{\min(\ell,n)})$ 
          & $\Omega(n_G2^{\min(\ell/2,n)})$ & $O(\ell n m^2 + n n_G m)$ \\
        '' ov. deg. $k$ %& $m=O(n^{k})$ 
          & '' & $\Omega(\ell n_G 2^{\min(\ell,n)})$ 
          & $\Omega(n_G2^{\min(\ell/2,n)})$  
          %& $O((\ell n  + n_G) m) = O((\ell n  + n_G) n^{k})$ \\
          & $O((\ell n  + \ell k 2^{k} + k n_G) n^{k})$ \\
        \hline
        $d$-multi-rng. %& $m=O((2n)^d)$ 
        & $O( \ell n n_G (2n)^d)$ &
          -- & $\Omega(\paren{\frac{n}{d}}^{d-1} n_G \frac{m}{d})$ 
          & $O(d n m^2 + n n_G m)$ \\ % d au lieu de \ell : constant size field
        '' ov. deg. $k$ %& $m=O(n^{k})$ 
          & '' & -- & $\Omega(\paren{\frac{n}{d}}^{d-1} n_G \frac{m}{d})$ 
        %$\Omega(\frac{2^d}{d} n_G m)$ 
          %& $O( (\log ^d n + n_G) m) = O( (\log ^d n + n_G) n^{k})$ \\
          & $O((\ell k^{d+1}\log^d n  + \ell k 2^{k} +  k n_G) n^{k})$ \\
        \hline
      \end{tabular}
      \egroup
    }
    \vspace{1mm}
    \caption{Worst-case complexity of forwarding loop detection with $n$ rules
      that generate $m$ header classes in an $n_G$-node network, for various
      rule set representations: $T_\ell$-bounded for intersection and
      cardinality computations in $O(T_\ell)$ time; $\ell$-wildcard for
      wildcard expressions with $\ell$ letters; $d$-multi-rng. for multi-ranges
      in dimension $d$ (with $\ell=O(d)$). Additional hypothesis
      ``ov. deg. $k$'' stands for overlapping degree of rule sets 
      bounded by $k$.
      %Main parameters are $n,m,n_G$ (number of rules, headers classes and nodes resp.).
      %% The ``Any'' row stands for any data-structure
      %% enabling two set intersection, set cardinal computation, and membership
      %% testing in $O(f(\ell)$ time for some function $f$.
      %% The parameters are the length $\ell$ of headers in
      %% bits, the number $n$ of distinct rules, the number $m\le 2^{\min(\ell,n)}$
      %% of header classes, the maximum overlapping degree $k\le n$ of the rules,
      %% the average overlapping degree $\k\le k$ of the rules, the average
      %% overlapping degree $\K\le \min(2^k,m)$ of the characteristic sets,
      %% the number $n_G$ of nodes, and the average number $\delta\le n$ of rules
      %% per forwarding table. 
      %% When rules are represented with arbitrary
      %% compact set data-structure, $f(\ell)$ denotes a time upper-bound for two
      %% sets intersection, cardinal computation, and membership testing.  In the
      %% case of multi-range matching, the number of fields in headers is denoted by
      %% $d$.
    }
    \label{tab:comp}
  \end{center}
\vspace{-10mm}
\end{table}

Second, we provide a dimension parameter, the \emph{overlapping degree} $k$,
that captures the complexity of a collection of rule sets considered in a
forwarding network. It is defined as the maximum number of distinct rules
(i.e. with pairwise distinct associated sets) that
match a given header. This parameter constitutes a measure of complexity
for the field of sets generated by a given collection of sets. In
the context of practical hierarchical networks, we have the following 
intuitive reason to
believe that this parameter is low:
in such networks, more specific rules are used at lower levels of the
hierarchy. We can thus expect that the overlapping degree is
bounded by the number of layers of the hierarchy.
Empirically, we observed a value within
$5-15$ for datasets with hundreds to thousands of distinct multi-field rules,
and $k=8$ for the collection of IPv4 prefixes advertised in BGP. 
% (at time of writing) and most header classes (96~\%) are
%contained by three prefixes at most. Morever, we believe that
%previous observations of ``linear fragmentation''~\cite{hsa,hsa-report} could
%result mainly from low overlapping degree. 
%
A constant overlapping
degree implies that the number of header classes is polynomially bounded,
giving a hint on why practical networks are tractable despite the
NP-completeness of the problem. 
In addition, the algorithm we propose is tailored to take advantage of low
overlapping degree $k$, even without knowledge of $k$. 
%% For networks with constant overlapping degree, loop
%% detection takes $O(nm+n_Gm)$ time. This is indeed polynomial with respect to
%% the size of the input for constant $k$ as the number $m$ of header classes is
%% then $n^k$ at most.  (However, it is not fixed-parameter tractable for
%% parameter $k$.)  
Table~\ref{tab:comp} provides a summary of the complexity
results obtained for loop detection depending on how the sets associated to
rules are represented. 
%% For comparison, it also contains lower-bounds for
%% NetPlumber and Veriflow, and the complexity of the trivial brute force
%% algorithm that consists in testing all possible headers. 
%% Note that with $\ell$-bit headers and $n$
%% distinct rules, the maximum number $m$ of header classes is $2^{\min(\ell,n)}$
%% (consider for example the $\ell$ wildcard rules
%% $1*^{\ell-1},*1*^{\ell-2},\ldots$).
%% In the case of
%% multi-range rules, a simple optimization of brute force consists in testing
%% only all possible combinations of the range bounds observed in the rules.
%% Note that the number of header classes is then $(2n)^d$ at most.
%% We
%% refine our complexity bounds with respect to two more parameters $\k$
%% and $\K$ that are both constant when $k$ is constant. The average overlapping
%% degree $\average{k}$ is obtained by averaging over all classes the number of
%% distinct rules that match headers of the class. We similarly define the average
%% overlapping degree $\K$ of combinations (see
%% Section~\ref{sec:overlappingDeg} for more details). We obviously have
%% $\average{k}\le k$ and parameter $\K$ always satisfy $\K\le 2^k$. In the
%% measurements we have made, $\average{k}$ and $\K$ are respectively in the range
%% $2-5$ and $5-15$.
Our techniques can be extended to handle write actions (partial writes of fixed
values such as field replacement) while maintaining polynomial guarantees.

Our algorithm for atom computation solves two difficult technical issues.
First, it manages to remain polynomial in the number $m$ of atoms even
though the number of sets generated by intersection solely can be 
exponential in $m$ with general rules.
Second, the use of cardinality computations allows
to avoid exponential blow-up (in contrast with previous work) but naturally
induces a quadratic $O(m^2)$ term in the complexity. However, we manage
to reduce it to $O(m)$ in the case of collections with logarithmic
overlapping degree (i.e. $k=O(\log n)$). 
Indeed, our algorithm then becomes linear in $m$ and polynomial in $n$, 
providing an algorithmic breakthrough towards efficient atom enumeration.

\textbf{Roadmap:}
Section~\ref{sec:model} introduces the model. Section~\ref{sec:atoms}
describes how to represent the atoms of a field of sets.
We state in Section~\ref{sec:algo}
our main result concerning atom computation and its implications for forwarding
loop detection
that give the upper bounds listed in Table~\ref{tab:comp}.
Section~\ref{sec:comparison}
gives more insight about the comparison of our results with previous works
and justifies the lower bounds presented in Table~\ref{tab:comp}.
Finally, Section~\ref{sec:persp} discusses some perspectives.

\section{Model}\label{sec:model} % -------------------------------------

\subsection{Problem}
We consider a general model of network where a \emph{network instance} $\mathcal{N}$ is characterized by:
\begin{itemize}
\item a graph $ G=(V,E) $ where each node $u \in V$ is a \emph{router} and has a \emph{forwarding table} $T(u)$.
\item a natural number $\ell$ representing the (fixed) bit-length of packet headers. 
%that we consider to be fixed. Thus the set of all headers is finite. It is denoted by $H$.
\end{itemize}
Let $H$ denote the set of all $2^\ell$ possible headers (all $\ell$-bit
strings).  Each forwarding table $T(u)$ is an ordered list of forwarding rules
$(r_1,a_1),\ldots,(r_p,a_p)$.  Each \emph{forwarding rule} $(r,a)$, is made of a
predicate filter $r$ and an action $a$ to apply on any packet whose header matches the predicate. We say that 
a header $h$ \emph{matches} rule $(r,a)$ when it matches predicate $r$
(we may equivalently say that $(r,a)$ matches $h$).
We then write $h\in r$ to emphasize the fact that $r$ can be viewed
as a compact data-structure encoding the set of headers that match it.
This set is called the \emph{rule set} associated to $(r,a)$.
For the ease of notation, we thus let $r$ denote both the 
predicated filter of rule $(r,a)$ and the associated set.
%% associated
%% to the rule
%% both denote the predicate
%% filter of the rule and the set of headers it represents.
%% We write $h\in r$ when a header $h$ \emph{matches} rule $(r,a)$.
%Note that $r.s$ is represented by the predicate filter associated to $r$.
% (we may also say that $r$ matches $h$).

We consider three possible actions for a packet: forward to a neighbor, drop, or deliver (when the packet has reached destination).  The
priority of rules is given by their ordering: when a packet with header
$h$ arrives at node $u$, the first rule matched by $h$ is
applied. Equivalently, the rule $(r_i,a_i)$ is applied when $h\in r_i\cap
\overline{r_1} \cap \cdots \cap \overline{r_{i-1}}$, where $\overline{r}$
denotes the complement of $r$.  When no match is found (i.e. $h\in
\overline{r_1} \cap \cdots \cap \overline{r_p}$), the packet is dropped.

Given a header $h$, the \emph{forwarding graph} $G_h=(V,E_h)$ of $h$ represents
the forwarding actions taken on a packet with header $h$: $uv\in E_h$ when the first rule that matches $h$ in $T(u)$ indicates to forward to $v$. The \emph{forwarding
loop detection} problem consists in deciding whether there exists a 
header $h\in H$ such that $G_h$ has a directed cycle.

Note that we make the simplifying assumption that the input port of 
an incoming packet is not taken into account in the forwarding decision of
a node. In a more general setting, a node has a forwarding table for each
incoming link. This is essentially the same model except that we consider the
line-graph of $G$ instead of $G$.

\subsection{Header Classes}\label{sec:hc}
A natural relation of equivalency exists between
headers with respect to rules: two headers are equivalent if they match exactly
the same rules, that is if they belong to the same rule sets.
Trivially, two equivalent headers let the corresponding packets
have exactly the same behavior in the network.
% (the reverse is not true,
%i.e. two headers following the same path in the network may not match
%similarly other rules not considered along the path). 
The resulting equivalence
classes partitions the header set $H$ into nonempty disjoint subsets called
\emph{header classes}. To check any property of the
network, it  suffices to do it on a class-by-class basis instead of a
header-by-header basis. The number $m$ of header classes is thus a
natural parameter when considering the difficulty of forwarding loop 
detection (or other similar network analysis problems).
A more accurate definition of classes could take into account the order
of rules and the topology (see Appendix~E) but this does not allows us
to obtain complexity gain in general.

The header classes can be defined according to the collection $\R$
of rule sets of $\mathcal{N}$ (i.e. $\R=\set{r\mid \exists u,a \mbox{ s.t. } (r,a)\in T(u)}$).
%, regardless of associated actions and network topology.
%(We call \emph{collection} a set of set).  
If $\R(h)\subseteq \R$ denotes the set of all rule sets associated to the rules matched by a given header $h$, then its header class is clearly equal to $\paren{\cap_{r\in \R(h)}r} \cap
\paren{\cap_{r \in \R \setminus \R(h)}\overline{r}}$ (with the convention $ \cap_{r\in \emptyset} r = H $).
% Note that we use the convention $ \cap_{r\in \emptyset} r = H $. 
%% Regardless of the topology of the network and actions associated to rules, 
%% we consider the whole collection of
%% rule sets denoted by $\R=\cup_{u \in V} T_s(u)=\{r_1, r_2,\ldots,r_n\}$
%% where $T_s(u)$ denotes the rule sets associated to the rules in $T(u)$.
Such sets are the atoms of the field of sets generated by $\R$.
Their computation is the main topic of this paper and is detailed in the
next section. 

% It is possible to refine the notion of header class to take into account
% the order of rules and the topology of the network as developed
% in Appendix~E. However, we still need to compute the atoms of the
% collection of rule sets encountered along a path for emptiness tests.
% (Avoiding such tests can result in exploring an exponential number of paths
% as shown in Appendix~B.2.)

\subsection{Set representation}\label{sec:repr}

As we focus on the collection $\R$ of rule sets, we now detail our hypothesis
on their representation.  We assume that a data-structure $\D$ allows to
represent some of the subsets of a space $H$.  For the ease of notation, $\D\subseteq \mathcal{P}(H)$ also denotes the collection of subsets that can
be represented with $\D$. We assume that $ \D $ is closed under intersection: if $s$ and $ s' $ are in $ \D $, so is $ s\cap s' $.
 We say that such a data-structure $\D$ for subsets
of $H$ is \emph{$T_H$-bounded} when intersection and cardinality can be computed in
time $T_H$ at most: given the representation of $s,s'\in\D$,
the representation of $s\cap s'\in\D$ and the size $\card{s}$ of $s$ (as a
binary big integer) can be computed within time $T_H$. As big integers computed within time $T_H$ have $O(T_H)$ bits, this implies $\card{H} = 2^{O(T_H)}$: the
bound $T_H$ obviously depends on $H$.  Intersection, inclusion test
($s\subseteq s'$), cardinality computation ($\card{s}$) 
and cardinality operations
(addition, subtraction and comparison) are called \emph{elementary set
  operations}. Under the $T_H$-bounded hypothesis, all these
operations can be performed in $O(T_H)$ time ($s\subseteq s'$ is equivalent to
$\card{s\cap s'}=\card{s}$).
% and the simple big integer operations considered can
%be performed in linear time).

Two typical examples of data-structures meeting the above requirements are
wildcard expressions and multi-ranges.  In a forwarding network, we consider
the header space $H=\set{0,1}^\ell$ of all $\ell$-bit strings, which may be
decomposed in several fields. A rule set is typically
represented by a wildcard expression or a range of integers. In both cases,
they can be represented within $2\ell$ bits and both representation are
$O(\ell)$-bounded.  We call \emph{$\ell$-wildcard} a string $e_1\cdots
e_\ell\in\set{0,1,*}^\ell$.  It represents the set $\{x_1\cdots
x_\ell\in\set{0,1}^\ell \mid \forall i, x_i=e_i \mbox{ or } e_i=*\}$. If rules
are decomposed into fields, any combinations of wildcard expressions and ranges
can be used (either one for each field) and represented within $O(\ell)$
bits. However cardinality computations can take $\Theta(\ell\log \ell)$ time as
multiplications of big integers are required.  Such representation is thus
$O(\ell\log \ell)$-bounded.  Given $d$ field lengths $\ell_1,\ldots,\ell_d$
with sum $\ell$, we call \emph{$(d,\ell)$-multi-range} a cartesian product
$[a_1,b_1]\times\cdots\times [a_d,b_d]$ of $d$ integer ranges with $0\le a_i\le
b_i<2^{\ell_i}$ for $i$ in $1..d$. It represents the set
$\set{bin(x_1,\ell_1)\cdots bin(x_d,\ell_d) \mid (x_1,\ldots,x_d)\in
  [a_1,b_1]\times\cdots\times [a_d,b_d]}$ where $bin(x_i,\ell_i)$ is the binary
representation of $x_i$ within $\ell_i$ bits.

When manipulating a collection of $p$ sets in $\D$, we assume
that their representations are stored in a balanced binary search tree,
allowing to dynamically add, remove or test membership of a set 
in $O(T_H\log p)$ time. More efficient data-structures
(tries and segment trees) can be used for
wildcard expressions and multi-ranges as detailed in Appendix~A.1.
% App. A.1 comes here

\section{Atoms and combinations generated by a collection of sets}
\label{sec:atoms}

Given a space of elements $H$, a \emph{collection} is a finite set of
subsets of $H$. The \emph{field of sets}
$\sigma(\R)$ generated by a collection $\R$ is the (finite) $\sigma$-algebra generated by $\R$, that is the smallest collection closed under intersection, union and complement that contains $\R\cup\set{\emptyset,H}$.

The \emph{atoms} of $\sigma(\R)$ are classically defined as the non-empty elements
that are minimal for inclusion.  For brevity, we call them the atoms generated by $\R$. Let $\A(\R)$ denote their collection.
%for the atoms of the field of sets generated by $\R$.
Note that for $a\in\A(\R)$ and $r\in\R$, we have either
$a\subseteq r$ or $a\subseteq \overline{r}$ (otherwise $a\cap r$ and
$a\cap\overline{r}$ would be non-empty elements of $\sigma(\R)$ strictly
included in $a$).  This gives a characterization of the
atoms that  matches our definition of header classes when $\R$ is the
collection of rule sets of a network (see \Cref{sec:hc}):
\begin{equation}
\label{eq:def_atoms}
\A(\R)=\set{a\not=\emptyset \mid \exists R\subseteq \R, 
	a = \paren{\cap_{r\in R}r} \cap \paren{\cap_{r \in \R \setminus R}\overline{r}}}\text{.}
\end{equation}
%Note that we rely on the convention $ \cap_{r\in \emptyset} r = H $.
%

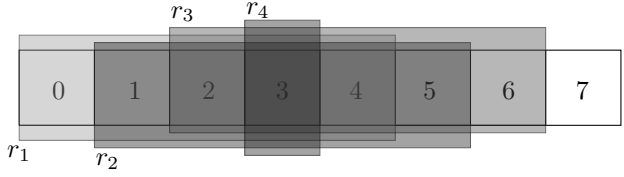
\begin{figure}
	\centering
	\begin{tikzpicture}
	\definecolor{dark}{gray}{0.1}
	\definecolor{yellow}{gray}{0.6}
	\definecolor{violet}{gray}{0.3}
	\definecolor{orange}{gray}{0.1}

	\draw (0,0) -- (0,1)-- (8,1)-- (8,0) -- (0,0);
	\draw (1,1) -- (1,0);
	\draw (2,1) -- (2,0);
	\draw (3,1) -- (3,0);
	\draw (4,1) -- (4,0);
	\draw (5,1) -- (5,0);
	\draw (6,1) -- (6,0);
	\draw (7,1) -- (7,0);
	
	\draw (0,0) -- (7,0)
	node[pos=0.075, yshift=13]{0} node[pos=0.22, yshift=13]{1} node[pos=0.36, yshift=13]{2} node[pos=0.5, yshift=13]{3} node[pos=0.64, yshift=13]{4} node[pos=0.78, yshift=13]{5} node[pos=0.93, yshift=13]{6} node[pos=1.07, yshift=13]{7};
	
	\draw[fill=yellow, opacity=0.4] (0,-0.2) -- (0,1.2)-- (5,1.2)-- (5,-0.2) -- (0,-0.2)
	node[yshift=-5, opacity = 1]{$r_1$};
	
	\draw[fill=orange, opacity=0.4] (1,-0.3) -- (1,1.1)-- (6,1.1)-- (6,-0.3) -- (1,-0.3)
	node[yshift=-5, xshift=5, opacity = 1]{$r_2$};
	
	\draw[fill=violet, opacity=0.4] (2,-0.1) -- (2,1.3)-- (7,1.3)-- (7,-0.1) -- (2,-0.1)
	node[yshift=45, xshift=5, opacity = 1]{$r_3$};
	
	\draw[fill=dark, opacity=0.4] (3,-0.4) -- (3,1.4)-- (4,1.4)-- (4,-0.4) -- (3,-0.4)
	node[yshift=55, xshift=5, opacity = 1]{$r_4$};
	
	\end{tikzpicture}
	\caption{Toy example of an 8-elements space $ H $ with a collection $ \R = \{r_1, r_2, r_3, r_4\} $. \label{fig:example}}
\vspace{-.5cm}
\end{figure}

%For example, the collection $ \R $ pictured in \Cref{fig:example} generates 6 atoms: 
%$ \{0\} = r_1 \cap \bar{r_2} \cap \bar{r_3}$, 
%$ \{1\} = r_1 \cap \bar{r_2} \cap r_3$, 
%$ \{2,3\} = r_1 \cap r_2 \cap r_3$, 
%$ \{4\} = \bar{r_1} \cap r_2 \cap r_3$,
%$ \{5\} = \bar{r_1} \cap r_2 \cap \bar{r_3}$, and
%$ \{6, 7\} = \bar{r_1} \cap \bar{r_2} \cap \bar{r_3}$.

%  $\A(\R)=\set{\{0\}, \{1\}, \{2,3\}, \{4\}, \{5\}, \{6,7\}}$.

For example, the collection $ \R $ pictured in \Cref{fig:example} generates 7 atoms: 
$ \{0\} = r_1 \cap \bar{r_2} \cap \bar{r_3} \cap \bar{r_4}$, 
$ \{1\} = r_1 \cap r_2 \cap \bar{r_3} \cap \bar{r_4}$, 
$ \{2,4\} = r_1 \cap r_2 \cap r_3 \cap \bar{r_4}$, 
$ \{3\} = r_1 \cap r_2 \cap r_3 \cap r_4$, 
$ \{5\} = \bar{r_1} \cap r_2 \cap r_3 \cap \bar{r_4}$, 
$ \{6\} = \bar{r_1} \cap \bar{r_2} \cap r_3 \cap \bar{r_4}$, and
$ \{7\} = \bar{r_1} \cap \bar{r_2} \cap \bar{r_3} \cap \bar{r_4}$.

Due to the \emph{complement} operations, the atoms can be harder to represent than the rules they are generated from. In \Cref{fig:example}, all rules are ranges, but the atom $ \{2,4\} $ is not.
When the \emph{intersection} operation can be computed and represented efficiently (see \Cref{sec:repr}), it is natural to consider the collection $\C(\R)\subseteq \sigma(\R)$ of \emph{combinations} defined as sets that can obtain by intersection from sets in $\R$:
\begin{equation}
\C(\R):=\set{c\not=\emptyset \mid \exists R\subseteq \R, 
	c = \cap_{r\in R}r}\text{.}
\end{equation}

In \Cref{fig:example}, there are 8 combinations: $ r_1 $, $ r_1\cap r_2 $, $ r_1\cap r_3 $, $ r_4 $, $ r_2 \cap r_3 $, $ r_3 $, $ H $ and $ r_2 $.

%  $ H $, , $ r_2 $, $ r_3 $, $ r_4 $, ,  and $ r_2\cap r_3 $.
%
% generates 7 atoms: 
%$ \{0\} = r_1 \cap \bar{r_2} \cap \bar{r_3} \cap \bar{r_4}$, 
%$ \{1\} = r_1 \cap r_2 \cap \bar{r_3} \cap \bar{r_4}$, 
%$ \{2,4\} = r_1 \cap r_2 \cap r_3 \cap \bar{r_4}$, 
%$ \{3\} = r_1 \cap r_2 \cap r_3 \cap r_4$, 
%$ \{5\} = \bar{r_1} \cap r_2 \cap r_3 \cap \bar{r_4}$, 
%$ \{6\} = \bar{r_1} \cap \bar{r_2} \cap r_3 \cap \bar{r_4}$, and
%$ \{7\} = \bar{r_1} \cap \bar{r_2} \cap \bar{r_3} \cap \bar{r_4}$.
%

%Before going further, we need a few additional definitions.
 Given a collection $\R$ and a subset $ s \in H $,  
we let $\R(s):=\set{r\in\R\mid s\subseteq r}$ denote 
the $ containers $ of $ s $, that is the sets in $\R$ that contain $s$. 
We associate  each combination $ c\in \C(\R) $ with the set 
$a(c):=c\cap\paren{\cap_{r\in\R\setminus\R(c)}\overline{r}}$.
The function $a(\cdot)$ (with parenthesis) 
should not be confused with an atom $a$ (without parenthesis).
A combination $ c $ is said to be \emph{covered} 
if $ a(c)= \emptyset $ (the union of its non-containers covers it), 
otherwise it is \emph{uncovered}. 
Similarly, we associate  each atom $ a \in \A(\R) $ with the combination $ c(a):= \cap_{r\in \R(a)}r $ (this corresponds to the ``positive'' part of the characterization  from Equation~\ref{eq:def_atoms}).
The following proposition states that atoms
%, which we need to address our problem, 
can be represented by uncovered combinations.

\begin{proposition}\label{prop:repr}
The collection $\UC(\R):=\set{c\in\C(\R)\mid a(c)\not=\emptyset}$ 
of uncovered combinations is in one-to-one correspondence with the atom collection $\A(\R)$: each $c\in\UC(\R)$ corresponds to the atom $ a(c) $. 
Reciprocally, each atom $a$ is mapped to combination $ c(a) $.
\end{proposition}

Based on \Cref{prop:repr}, we say that an uncovered combination $ c $ \emph{represents} atom $ a(c) $. $\UC(\R)$ can be seen as a canonical representation of $\A(\R)$ by a collection of combinations.
%
% in the sense that it is the unique collection of combinations that inclusion-wise
%represent all atoms generated by $\R$.
%
In Figure \Cref{fig:example}, atom $ \{2,4\} $ is represented by 
combination $ r_1\cap r_3 = r_1\cap r_2 \cap r_3$. 
Similarly, atoms
%\Cref{prop:repr} states that the atoms 
$ \{0\} $, $ \{1\} $, %$ \{2,4\} $, 
$ \{3\} $, $ \{5\} $, $ \{6\} $ and $ \{7\} $
are represented by (uncovered) combinations 
$ r_1 $, $ r_1\cap r_2 $, %$ r_1\cap r_3 $, 
$ r_4 $, $ r_2 \cap r_3 $, $ r_3 $, 
and $ H $ respectively.
Combination $ r_2 $ is covered: $ a(r_2) = \bar{r_1}\cap r_2 \cap \bar{r_3}
\cap \bar{r_4} = \emptyset $.

The proof of \Cref{prop:repr} is straightforward. First, we verify that if $ c $ is an uncovered combination, $ a(c) $ is an atom: it suffices to observe that $ c = \cap_{r\in \R(c)} r $ and to match $a(c)=c\cap\paren{\cap_{r\in\R\setminus\R(c)}\overline{r}}$ with the atom characterization given by Equation~\ref{eq:def_atoms} using $ R = \R(c) $. 
%Using the same characterization and observing that $ \R(a(c)) = \R(c) $, we get $ c(a(c)) = c$. 
Similarly, if $ a = \paren{\cap_{r\in R}r} \cap \paren{\cap_{r \in \R \setminus R}\overline{r}} $  is an atom for some $ R \subseteq \R $, the combination $ c(a) $ satisfies $ \R(c(a)) = R $ which implies $ a(c(a)) = a$. In particular, as $ a\neq \emptyset $, $ c(a) $ is uncovered.

%Note that thanks to the construction of the representative combinations (an
%uncovered combination corresponds to the ``positive'' part of the
A nice property of the characterization of \Cref{prop:repr} is that
it allows to efficiently test 
whether a set $r\in\R$ contains an atom $a\in\A(\R)$: 
given a combination $c$ that represents $a$, $a\subseteq r$ is equivalent to
$c\subseteq r$. This comes from the fact that every uncovered combination $c$
has same containers as $a(c)$. (If it was not the  case, $a(c)=\emptyset$
and $c$ is covered.)
This explains the importance of determining $ \UC(\R) $, and in
particular to separate covered combinations from uncovered ones.
This is the subject of the next section, but we first formally introduce
the notion of overlapping degree.

\paragraph*{Overlapping degree of a collection}
Our representation is naturally associated to the following measure 
of complexity of a collection $\R$. We define the \emph{overlapping degree} $k$
of $\R$ as the maximum number of containers of an element, that is
$k=\max_{h\in H}\card{\R(\set{h})}$. Note that all elements within an atom
have same containers in $\R$ and that a set cannot have more containers than
any of its elements. We thus have:
\[
k = \max_{s\subseteq H}\card{\R(s)} = \max_{a\in\A(\R)}\card{\R(a)}
\]
As any atom $a$ can be expressed as 
$a=\paren{\cap_{r\in R}r} \cap \paren{\cap_{r \in \R \setminus R}\overline{r}}$
where $R=\R(a)$ is  the set of containers of $a$, the number of atoms
is obviously bounded by $n \choose k$ where $n=\card{\R}$ denotes the number
of sets in $\R$. The overlapping degree of a collection thus measures 
its complexity in terms of number of atoms it may generate.

We similarly define the \emph{average overlapping degree} $\k$ as
the average number of containers of an atom: 
$\k = \frac{\sum_{a\in\A(\R)}\card{\R(a)}}{\card{\A(\R)}}$.
We obviously have $\k\le k$.
Given a collection $\R$, we will also consider the average overlapping degree $\K$
of combinations, that is the average overlapping degree of $\C(\R)$.
Note that $\C(\R)$ has the same
collection of atoms as $\R$ and that a combination $c=\cap_{r\in \R(c)}r$
containing an atom $a$ must satisfy $\R(c)\subseteq \R(a)$. The overlapping
degree of $\C(\R)$ is thus at most $2^{k}$ and we always have $\K\le 2^{k}$.

In the example from \Cref{fig:example}, one can verify that we have $ k = 4$, $ \k = 13/7 $ and $ \K = 4 $.

In real datasets we observe that both $k$ and $\K$ are in the range $[2,15]$
while $\k$ is in the range $[1.5,5]$
(these include Inria firewall rules, Stanford forwarding tables provided
by Kazemian et al.~\cite{hsarepo} and IPv4 prefixes announced at BGP level from
Route Views~\cite{routeviews}).

\section{Incremental computation of atoms}\label{sec:algo}

We can now state our main result concerning the computation of the atoms
generated by a collection of sets.
% (See Subsection~\ref{sec:repr} for the detail of the various set operations considered.)

%One can easily see that this bound is tight by considering
%$\R=\set{\overline{\set{1}},\ldots,\overline{\set{n}}}$ with $H=\set{1..n}$.

\begin{theorem}\label{th:main}
Given a space set $H$ and
a collection $\R$ of $n$ subsets of $H$, the collection $\UC(\R)$ 
of combinations that canonically 
%inclusion-wise 
represent the atoms $\A(\R)$
can be incrementally computed with $O(\min(n+k\K\log m, n m) m)$ 
elementary set operations where:
$m$ is the number of atoms generated by $\R$; $k$ is the overlapping
degree of $\R$; $\K$ is the average overlapping degree of $\C(\R)$; 
$\k$ is the average overlapping degree of $\R$.
Within this computation, each combination $c\in\UC(\R)$ can be associated
to the list $\R(c)$ of sets in $\R$ that contain $c$.
If sets are represented by $\ell$-wildcard expressions (resp.
$(d,\ell)$-multi-ranges), the representation can be computed in
$O(\ell \min(n+k\K, n m) m)$ 
(resp. $O(\ell \min(\k \log^d m + k\K, n m) m)$) 
time. 
%% More precisely, if the data-structures used for representing sets 
%% and collections of sets enable elementary set operations within time $T_{set}$,
%% $p$-collection operations within time $T_{coll}(p)$ and $p$-intersection
%% queries with overhead $(T_{inter}(p),T_{updt}(m))$, then the representation
%% of the atoms generated by $\R$ can be computed in 
%% $O(n T_{inter}(m)  + \average{k} m T_{updt}(m)  
%%   + \min(k\K,\average{k}m) m (T_{set} + T_{coll}(m)))$ time.
\end{theorem}

\paragraph*{Application to forwarding loop detection}
%We first derive the consequences of this result for forwarding
%loop detection: given a representation of the atoms generated by the rule sets of a
%network, we have the following result for forwarding loop detection.
Theorem~\ref{th:main} has the following consequences for forwarding loop
detection.

\begin{corollary}\label{cor:loop}
Given a network $\mathcal{N}$ with collection $\R$ of $n$ rule sets with
$T_\ell$-bounded representation, forwarding loop detection can be performed in
$O(T_\ell \min(n+k\K\log m, n m) m + \k n_G m)$ time where $m$ is the number
of atoms in $\A(\R)$, $k$ is the overlapping degree of $\R$ and $\k$
(resp. $\K$) is the average overlapping degree of $\R$ (resp. $\C(\R)$).
If sets are represented by $\ell$-wildcard expressions (resp.
$(d,\ell)$-multi-ranges), the representation can be computed in
$O(\ell \min(n+k\K, n m) m  + \k n_G m)$
(resp. $O(\ell \min(\k \log^d m + k\K, n m) m + \k n_G m)$) 
time. 
\end{corollary}

This result can be extended to handle write actions as explained in Appendix~F
due to the lack of space.
The upper-bounds for forwarding loop detection listed in Table~\ref{tab:comp}
follow from Corollary~\ref{cor:loop} which is proved in Appendix~A.3. A key
ingredient consists in maintaining for each rule set a list that describes its
presence, priority and action for each node. Detecting a loop for a header
class $a$ then consists in merging the lists associated to the rule sets
containing $a$ (as provided by our atom representation) for obtaining the
forwarding graph $G_a=G_h$ for all $h\in a$. Directed cycle detection is
finally performed on each such graph.

In the rest of this Section, we prove Theorem~\ref{th:main} by introducing two incremental algorithms for atom computation and
analyzing their performance. The incremental approach allows to avoid exponential blow-up
even in cases where the number of combinations can be exponential in
the number $m$ of atoms (an example is given in Appendix~B.2).

% for updating the
%atoms generated by a collection $\R$ when a new set is added to $ \R $.

\paragraph*{Algorithms for updating atoms}
%
%Lemmas \ref{lem:addrule} and \ref{lem:cardinals} enable dynamic algorithms
%for updating the collection $\UC$ of uncovered combinations of a
%collection $\R$. 
We first propose a basic algorithm for updating
the collection $\UC$ of uncovered combinations of a collection $\R$ when
a set $r$ is added to $\R$.
The main idea is that  after adding $r$ to $\R$,
the only new uncovered combinations that can be created are intersections of
pre-existing uncovered combinations with $ r $
(see Lemma~\ref{lem:addrule} in Appendix~A.2).
%Following Lemma~\ref{lem:addrule}, 
We thus first add to $\UC$ the combinations $c\cap r$ for
$c\in \UC$. As this may introduce covered combinations,
we then compute the atom size $c.atsize=|a(c)|$ for each combination $c$.
It is then sufficient to remove any combination $c$ with $c.atsize=0$
to finally obtain $\UC(\R\cup\set{r})$.
This atom size computation is possible because we have
$d=\cup_{c\in \UC \mid c \subseteq d}a(c)$ for all $d\in\UC$.
% in any collection $\C'$ such that $\UC(\R)\subseteq \C'\subseteq \C(\R)$.
(see Lemma~\ref{lem:cardinals} in Appendix~A.2).
As this union is disjoint, we have 
$|a(d)|=|d| - \sum_{c\in \UC \mid c \subsetneq d} |a(c)|$.
We thus compute the inclusion relation between combinations 
and store in $c.sup$ the combinations that strictly contain $c$.
Initializing $c.atsize$ to $|c|$ for all $c$, we then
scan all combinations $c$ by non-decreasing cardinality
(or any topological order for inclusion) and subtract $c.atsize$
from $d.atsize$ for each $d\in c.sup$.
A simple proof by induction allows to prove that $c.atsize = |a(c)|$
when $c$ is scanned.
The whole process in summarized in Algorithm~\ref{algo:basic}.

\smallskip
\begin{algorithm}[H]
    \Procedure{$\mb{basicAdd}(r, \R, \UC )$}{
      $\UC' := \UC \cup \set{c\cap r \mid c\in \UC}$\;
      \ForEach{$c\in\UC'$}{
        $c.atsize := |c|$\;
        $c.sup := \set{d\in \UC'\mid c\subsetneq d}$\;
      }
      \ForEach{$c\in\UC'$ in non-decreasing cardinality order}{
        \lForEach{$d\in c.sup$}{
          $d.atsize := d.atsize - c.atsize$\;
        }
      }
      $\UC := \UC' \setminus \set{c\in\UC' \mid c.atsize = 0}$\;
    }
\caption{Add a set $r$ to a collection $\R$ and
  update the collection $\UC= \UC(\R)$ of its uncovered combinations
  accordingly.}
\label{algo:basic}
\end{algorithm}
\smallskip

The correctness of Algorithm~\ref{algo:basic} follows from the two above
remarks (that is Lemma~\ref{lem:addrule} and Lemma~\ref{lem:cardinals} in
Appendix~A.2).  Its main complexity cost comes from intersecting $r$
with each combination and computing the inclusion relation between
combinations, that is $O(nm)$ and $O(m^2)$ elementary set operations
respectively. Starting from $\UC=\set{H}$
and incrementally applying Algorithm~\ref{algo:basic} to each set in $\R$
thus allows to obtain $\UC(\R)$ with $O(nm^2)$ elementary set operations,
yielding the general bound of Theorem~\ref{th:main}.

%\bigskip

To derive better bounds for low overlapping degree $k$, we propose
a more involved algorithm that maintains
$c.sup$ and $c.atsize$ from one iteration to another and
makes only the necessary updates. This requires to handle several subtleties 
to enable lower complexity. 

We similarly start by computing
the collection $\mb{Inter}=\set{c\in \UC\mid c\cap r\not=\emptyset}$ 
of combinations intersecting $r$. A first subtlety comes from the fact
that several combinations $c$ may result in the same $c'=c\cap r$.
However, we are only interested in the combination $c$ which is minimal
for inclusion that we call the \emph{parent} of $c'$. The reason is that
$c'.sup$ can then be computed from $c.sup$. The parent is unique unless $c'$
is covered in which case $c'$ is marked as covered and discarded
(see the argument for $c'.sup$ computation later on).
To obtain right parent information, we thus process all $c\in\mb{Inter}$ 
by non-decreasing cardinality. The produced combinations $c'=c\cap r$
such $c'$ were not in $\UC$ are called \emph{new} combinations.
Their atom size is initialized to $c'.atsize = |c'|$.
See the ``Parent computation'' part of 
Algorithm~\ref{algo:add}.

We then remark that we only need to compute (or update) $c.sup$ 
for combinations that include $r$, which we store in a set $\mb{Incl}$.
We also note that $c.atsize$ needs to be computed when $c$ is new and
updated when $c$ is the parent of a new combination.
A second subtlety resides in computing (or updating) $c.sup$ only when
$c$ is not covered that is when $c.atsize$ (after computation) 
appears to be non-zero. As the computation of $c.sup$ lists is the most
heavy part of the computation, this is necessary to enable our complexity
analysis. For that purpose, we scan $\mb{Incl}$ by 
non-decreasing cardinality so that the correct value of $c.atsize$
is known when $c$ is scanned similarly as in Algorithm~\ref{algo:basic}.
However, we avoid any useless computation when $c.atsize$ is zero.
Otherwise, we compute 
(or update) $c.sup$ and decrease $d.atsize$ by $c.atsize$
from $d\in c.sup$ for adequate $d$:
if $c$ and $d$ where both in $\UC$, this computation has already been made;
it is only necessary when $d$ is new or when $d$ is the parent of $c$.
We optionally maintain for each combination
$c$ a list $c.cont$ that contain the list of sets $r\in\R$ that contain $c$
(Such lists are not necessary for the computation but they are useful
for loop detection as detailed in Appendix A.3).
See the ``Atom size computation'' part of 
Algorithm~\ref{algo:add}.

\begin{algorithm}[!ht]
  {\small
    \Procedure{$\mb{add}(r, \R, \UC = \UC(\R))$}{
      $\mb{New} := \emptyset$; $\mb{Incl} := \emptyset$;\;
      \Comment{ ---------------------- Parent computation ---------------------- }
      $\mb{Inter} := \set{ c\in \UC \mid c\cap r\not=\emptyset}$ \;
      Sort $\mb{Inter}$ by non-decreasing cardinality.\;
      \ForEach{$c\in\mb{Inter}$}{
        $c ':= c\cap r$\;
        \eIf{$c'\notin \mb{Incl}$}{
          \If{$c'\notin \UC$}{
            $\UC := \UC\cup\set{c'}$; $\mb{New} := \mb{New}\cup\set{c'}$;\;
            $c'.atsize := |c'|$; $c'.sup := \set{}$; $c'.cont := \set{}$
            \hspace{1mm}
            \Comment{Updated later.}
          }
          $c'.parent := c$; $c'.covered := false$; 
          $\mb{Incl} := \mb{Incl} \cup \set{c'}$;\;
        }{
        \lIf{$c'.parent\not\subseteq c$}{$c'.covered := true$}
      }
    }
    Remove from $\mb{Incl}$, $\mb{New}$ and $ \UC $
    all $c$ such that $c.covered = true$.\;
    \Comment{ ---------------------- Atom size computation ---------------------- }
    Sort $\mb{Incl}$ by non-decreasing cardinality.\;
    \ForEach{$c\in\mb{Incl}$}{
      \If{$c.atsize > 0$}{
        \Comment{Adjust $ c.sup $, $c.cont$ 
          and update $d.atsize$ for impacted $d\supsetneq c$:}
        %    \Comment{Update $c.sup = \set{d\in \UC\mid c\subsetneq d}$
        %            and $d.atsize$ for $d\in c.sup$:}
        \If{$c\in\mb{New}$}{
          $c.parent.atsize := c.parent.atsize - c.atsize$\;
          $c.sup := \set{c.parent}\cup c.parent.sup$\; 
          $c.cont := c.parent.cont$\;
        }
        %      
        %      \lIf{$c\in\mb{New}$}{
        %        $c.sup := \set{c.parent}\cup c.parent.sup$
        %      }\;
        $c.sup := c.sup \cup
        \set{d\cap r\mid d\in c.sup
          \mbox{ and }d\cap r\in\mb{Incl}\setminus\set{c}}$\;
        $c.cont := c.cont\cup\set{r}$\;
        \ForEach{$d\in c.sup$ s.t. $d\in\mb{New}$}{
          $d.atsize := d.atsize - c.atsize$\;
        }
      }
    }
    \Comment{ ---------------------- Remove covered combinations ---------------------- }
    \ForEach{$c\in\mb{Incl}$}{
      Remove from $c.sup$ any $d$ such that $d.atsize=0$.\;
      \lIf{$c.atsize = 0$}{
        $\UC := \UC\setminus\set{c}$; 
        $\mb{Incl} := \mb{Incl}\setminus\set{c}$; 
      }\;
      \lIf{$c\in\mb{New}$ and $c.parent.atsize = 0$}{
        $\UC := \UC\setminus\set{c.parent}$
      }
    }
    %  \Return{$\mb{Incl}$}
  }
}
\vspace{-1mm}
\caption{Add a set $r$ to a collection $\R$ and
  update the collection $\UC$ of its uncovered combinations
  accordingly.}
\label{algo:add}
\end{algorithm}

A last critical point resides in the computation of $c'.sup$ for
each new combination $c'$.
The $c'.sup$ list can be obtained from $c'.parent.sup$ 
by copying and also intersecting elements of $c'.parent.sup$ with
$r$. This is sufficient: for $d\in\UC$ such that $c'\subsetneq d\cap r$,
we can consider $c=c'.parent\cap d$. If $c\in\UC$,
$c'.parent = c$ by minimality of $c'.parent$ and
we thus have $d\in c'.parent.sup$. The case where $c\notin\UC$ and
$d\notin c'.parent.sup$ cannot happen as it would imply that two
different combinations $c_1=c'.parent$ and $c_2\subseteq c$ generate 
$c'$ by intersection with $r$ ($c_1\cap r=c_2\cap r=c'$) and are both minimal
for inclusion. In such case, $c_1\cap c_2$ was covered in $\R$ and 
so would be $c'$ in $\R$ (and also in $\R\cup\set{r}$). That is why such
combination $c'$ are already discarded during parent computation.
On the other hand, the list $c.sup$ of
a combination $c\in\UC$ can be updated by intersecting
elements of $c.sup$ with $r$: when $c\subsetneq d\cap r$ for $c\subseteq r$,
we have $c\subsetneq d$.

Finally, combinations $c$ with $c.atsize=0$ are discarded and removed from
$b.sup$ list of remaining combinations as detailed in the
``Remove covered combinations'' part of
Algorithm~\ref{algo:add}. 

\medskip

The main argument in the complexity analyzis of Algorithm~\ref{algo:add}
comes from bounding the size of each $c.sup$ list by $|\UC(a(c))|$
(this is where the overlapping degree of $\C(\R)$ is used). 
The number of elementary set operations performed is indeed
$O(m+\sum_{a\in A_r}\card{\UC'(a)}\log m)$ where $m$ denotes the number of
atoms of $\R$, $A_r=\set{a\in\A(\R') \mid a\subseteq r}$ denotes the atoms of
$\R'=\R\cup\set{r}$ included in $r$, and $\UC'(a)$ denotes the uncovered
combinations of $\R'$ that contain $a$.  A key step consists in proving
$|\mb{Inter}|\le \sum_{a\in A_r}\card{\UC'(a)}$.  
The $\log m$ accounts for membership tests and add/remove operations
on collections of combinations.
In the case of
$\ell$-wildcard and $(d,\ell)$-multi-ranges, this factor can be saved
by storing collections of sets in tries rather than balanced binary search
trees.  With $(d,\ell)$-multi-ranges, a segment tree allows to retrieve
$\mb{Inter}$ in $O(|\mb{Inter}| + \log^d m)$ elementary set
operations~\cite{BergGeomBook,intersectionBoxes}.

The various bounds for low overlapping degree
of Theorem~\ref{th:main} results from applying
iteratively Algorithm~\ref{algo:add} for
each set of $\R$ and by carefully bounding the sums of $\card{\UC'(a)}$
terms. Intuitively, each atom $a$ generated by the collection $\R_i$
of the first $i$ sets of $\R$ can be associated to an atom $a'$ of $\R$ 
such that $\card{\UC'(a)}\le \card{\C(a')}$ where $\C=\C(\R)$.
As each atom is included in $k$ sets of $\R$ at most, this allows to
bound the overall sum of $\card{\UC'(a)}$ terms by $k\K m$. 
Details are given in Appendix~A.2.

\section{Comparison with previous works}\label{sec:comparison} % ------------

%As far as we know, all previous works rely on complement computations
%(or equivalent operations).
%The complement of a single set generally requires to be represented as a union
%of several intermediate sets (up to $\ell$ for $\ell$-wildcards and up to
%$2d-1$ for $d$-multi-ranges). 
We give in Appendix~C
more details about ``linear fragmentation'', a phenomenon observed by 
Kazemian et al.~\cite{hsa-report}, and propose small overlap
degree as a plausible cause.
The notion of uncovered combination
is linked to that of weak completion 
introduced by~\cite{boutier} in the context of rule-conflict resolution
as detailed in Appendix~D.
We now provide examples where use of set complement computation can lead
to exponential blow-up in previous work. 

%% As previous work rely on computations of complement of sets, more assumptions
%% on the data-structure used for set representation is necessary. This tends to 
%% be costly as the complement of single set requires generally to be
%% represented as the union of several intermediate sets (up to $\ell$ for
%% wildcards, and up to $2d-1$ for $d$ dimensional multi-ranges).  Although
%% Veriflow~\cite{veriflow} and HSA/NetPlumber~\cite{hsa,netplumber} behave
%% generally well in practice, there are case where the number of such
%% intermediate sets generated can explode exponentially in $\min(\ell,n)$ for
%% wildcards and in $d$ for multi-ranges. This is true even for networks with
%% overlapping degree as low as 2. This phenomenon comes from the fact the
%% intermediate sets generated by complement operations may increase the
%% overlapping degree up to $\ell$ for wildcards and up to $d$ for multi-ranges.
%% We provide lower-bounds on the worst case complexity of both HSA/NetPlumber and
%% Veriflow showing that their complexity can be exponential in $\min(\ell,n)$ for
%% wildcard matching in networks where the number $m$ of header classes is linear
%% in $n$. Such cases are constructed with simple disjoint rules and could occur
%% in practice.

\subsection{Veriflow}\label{sec:exp-verif}

Veriflow~\cite{veriflow} incrementally computes a partition into sub-classes
that forms a refinement of the header classes: when a rule $r$ is added, each
sub-class $c$ is replaced by $c\cap r$ and a partition of $c\setminus r$.
%Our computation of uncovered combinations is thus part of the sub-classes
%computed by Veriflow. For constant overlapping degree, our cardinality
%computations cost $O(\ell)$ per combination and our method is then
%always competitive compared to Veriflow. 
%We now show that this representation can grow exponentially with $d$,
%even with low number of header classes.
%sometimes offer a speed-up of $\frac{2^d}{d}$ where $d$ denotes the number
%of header fields.
%
Veriflow benefits from the hypothesis that
 headers can be decomposed into  $d$ fixed fields
and that each rule set can be represented by a multi-range
$r=[a_1,b_1]\times\cdots\times[a_d,b_d]$. 
%% Such a multi-range $r$ 
%% represents the set of headers
%% whose first field (read as a binary integer) is in $[a_1,b_1]$,
%% and whose second field is in $[a_2,b_2]$, and so on.
%% Similarly, each sub-class $s$ is represented by a multi-range.
The intersection of two multi-ranges is obviously a multi-range.
However, set difference is obtained by intersection with the complement
%The complementary of $r$
which is represented as the union of up to $2d-1$ multi-ranges.
%% \begin{itemize}
%% \item $[0,a_1-1]\times H_{2..d}$ and $[b_1+1,\infty_1]\times H_{2..d}$,
%% \item $[a_1,b_1]\times [0,a_2-1]\times H_{3..d}$ and
%%       $[a_1,b_1]\times [b_2+1,\infty_2]\times H_{3..d}$,
%% \item $\cdots$,
%% \item $[a_1,b_1]\times \cdots\times [a_{d-1},b_{d-1}]\times [0,a_d-1]$ and
%%       $[a_1,b_1]\times \cdots\times [a_{d-1},b_{d-1}]\times [b_d+1,\infty_d]$,
%% \end{itemize}
%% where $\infty_i$ denotes the maximum
%% possible value in field $i$,
%% and $H_{i..j}=[0,\infty_i]\times\cdots\times [0,\infty_j]$ denotes the
%% multi-range of all possible values for fields $i,\ldots,j$ for
%% $1\le i \le j \le d$.
%The set difference $c\setminus r$ can thus be represented by
%$2d-1$ multi-ranges at most.
%
%We now provide an example where the number of sub-classes generated by Veriflow
%can increase exponentially with the number of fields $d$. For each field $i$,
%consider two values $0<a_i<b_i<\infty_i$ and consider the rules:
%We construct a difficult input for Veriflow as follows.
We prove in Appendix~B.1 that successive difference with rules of the form 
$[O,\infty]^{i-1}\times [a_j,a_j] \times [b,b]^{d-i}$ with
$0 < a_1,\ldots,a_p < b$ can generate
$\Omega(\paren{\frac{n}{d}}^d\frac{m}{d})$ multi-ranges in the
computations of Veriflow as indicated in Table~\ref{tab:comp}.
%% As it has overlapping degree 2, this justifies the two lower-bounds
%% indicated for Veriflow in Table~\ref{tab:comp} for $d$-multi-ranges.
%% (The lower-bounds for $\ell$-wildcards are similarly obtained
%% in Appendix~B.1.)
%% For integral $p$ and $\ell\ge d\ceil{\log 2p}$, we consider that headers
%% are decomposed in $d$  equal-length fields. Each field can thus 
%% encode integral values up to $b=2p$ at least and we can consider
%% $p$ integral values 
%% $0 < a_1,\ldots,a_p < b$ such that $a_i+1<a_{i+1}$ for 
%% $i=0..p-1$. The collection consists in
%% the $n=dp+1$ following multi-ranges:
%% \[
%%   r_0=H_{1..d}, \hspace{5mm} \mbox{and} \hspace{5mm}
%%   r_i^j=H_{1..i-1} \times [a_j,a_j] \times [b,b]^{d-i}
%%    \mbox{ for } i,j\in [1..d]\times [1..p],
%% \]
%% where $H_{i..j}=[0,\infty_i]\times\cdots\times [0,\infty_j]$ denotes the
%% multi-range of all possible values for fields $i,\ldots,j$.
%% As the sets $r_i^j$
%% for $i,j\in [1..d]\times [1..p]$ are pairwise disjoint, the number of
%% header classes is $m=dp+1=n$.  

\subsection{HSA / NetPlumber}\label{sec:netpl-exp}

HSA/NetPlumber~\cite{hsa,netplumber} use clever heuristics to efficiently 
compute the set of headers $H_P$ than can traverse a given path $P$. 
An important one
consists in lazy subtraction: set difference computations are postponed
until the end of the path. For that purpose, this set $H_P$ is represented
as a union of terms of the form $s=c_0\setminus \cup_{i=1..p} c_i$ where
the elementary sets $c_0,\ldots,c_p$ are represented with wildcards. 
The emptiness of such terms is regularly tested.
A simple heuristic is used during the construction
of the path: $s$ is obviously empty if $c_0$ is included in $c_i$ for some $i\ge 1$. But if the path loops,  HSA has to develop the corresponding
terms into a union of wildcards to determine if one of them may produce a
forwarding loop.

We now provide an example where this emptiness test can take exponential time.
Consider a node whose forwarding table consists in $\ell+1$ rules with
following rule sets:
\[
  r_0=1^\ell, \hspace{5mm}
  r_i=1^{\ell-i}0*^{i-1} \mbox{ for } i=1..\ell, \hspace{5mm} \mbox{and} \hspace{5mm}
  r_{\ell+1}=*^\ell.
\]
All rules are associated with the drop action except the last rule
(with rule set $r_{\ell+1}$) whose action is to forward to the node itself.
Such a forwarding table is depicted in Figure~\ref{fig:one_node_loop} for
$\ell=4$. Starting a loop detection from that node, HSA detects
a loop for headers in $r_{\ell+1}\setminus\cup_{i=0..\ell}r_i$.
The emptiness of this term is thus tested. For that purpose, HSA
represents the complement of $r_i$ with
$0*^{\ell-1}\cup *0*^{\ell-2}\cup\cdots\cup *^{\ell-i-1}0*^i\cup *^{\ell-i}1*^{i-1}$.
Note that each of the $\ell-i$ wildcard expressions in that union have only one
non-$*$ letter.
Distributivity is then used to compute $r_{\ell+1}\setminus\cup_{i=0..\ell}r_i$
as $\overline{r_0}\cap \cdots \cap \overline{r_\ell}$. After expanding the
first $j-1$ intersections, HSA thus obtains a union of wildcards with
$j$ letters in $\set{0,1}$ and $\ell-j$ letters equal to $*$ that has
to be intersected with $\overline{r_{j+1}}\cap \cdots \cap \overline{r_\ell}$.
In particular, this unions contains all strings with $j$ letters equal to $0$
and $\ell-j$ equal to $*$. All $\ell$-letter strings with alphabet 
$\set{*,0}$ are produced during the computation which thus
requires $\Omega(\ell 2^\ell)$ time.
For testing a network with $n_G$ similar nodes, HSA thus requires time
$\Omega(\ell n_G 2^\ell)$. As all sets $r_0,\ldots,r_\ell$ are pairwise
disjoint, the overlapping degree of the collection is $k=2$ and this
justifies the two lower-bounds indicated for NetPlumber in
Table~\ref{tab:comp}.

%% \begin{figure}
%% \includegraphics[width=\textwidth]{hsa_loop}
%% \caption{A network whose verification is exponential with HSA.}
%% \label{fig:hsa_loop}
%% \end{figure}

\section{Future work}\label{sec:persp} % ---------------------
%\vspace{-.3cm}

Our approach could be
naturally integrated in the Veriflow~\cite{veriflow} 
framework both for speed-up and
performance-guarantee considerations. %Our model does not include write
%actions. However, we believe that our ideas could
Our ideas can also be integrated in the
NetPlumber~\cite{netplumber} framework (see the similar approach 
proposed in Appendix~E). 
%% Updating a collection of uncovered combinations can be done in persistent
%% style for managing efficiently several collections as they follow different 
%% paths. Write operations could be recorded in a specific wildcard expression
%% that serves as a general mask for all wildcards in the collection, allowing 
%% to apply efficiently such ``network transfer function'' (using
%% HSA/NetPlumber terminology) to a collection.
This would allow to enhance the
emptiness tests performed within NetPlumber to guarantee polynomial time
execution when the number of header classes is polynomially bounded.
In the context of multi-ranges, the emptiness test of expressions
is equivalent to the Klee measure problem which consists
in computing the volume of a union of boxes.  Indeed, expression
$r_p\setminus\cup_{i=1..p-1}r_i$ is empty when the volume of the union of the
boxes $r_1\cap r_p,\ldots,r_{p-1}\cap r_p$ equals that of $r_p$.  According to
recent work~\cite{easyKlee}, the complexity of this problem is believed to be
$\Theta(n^{d/2})$. It would be interesting to determine if such low complexity
bounds extend to atom computation in the case of multi-ranges rules.

\bibliography{checking_forwarding}

\newpage
\section*{Appendix A : Proof details of Theorem~\ref{th:main} and Corollary~\ref{cor:loop}}

We present here a refined version of Theorem~\ref{th:main}
with more hypothesis about the data-structure used to store a collection
of sets. We first review these hypothesis and then the proof details of
both versions.

\subsection*{A.1 Collection of sets representation}
% goto App. A.1

When manipulating a collection of $p$ sets in $\D$, we assume that their
representations are stored in a collection data-structure allowing to
dynamically add, remove or test membership of a set.  Such operation are called
\emph{$p$-collection operations}.  We can use a balanced binary search tree
when comparisons according to a total order can be performed. Such comparison
can usually be obtained by comparing directly the binary representations
themselves of the set in linear time (and thus $O(T_H)$ for sets with
$T_H$-bounded representation).
It is considered as an elementary set operation.
In the case of wildcard expressions, the complexity of these operations can be
reduced to $O(\ell)$ time by using a trie or a Patricia tree. Our algorithms
will also make use of an operation similar to stabbing query that we call
\emph{$p$-intersection} query. It consists in producing the list $L_s$ of sets
in a collection $\R$ of $p$ sets that intersect a given query set $s$
($L_s=\set{r\in \R\mid s\cap r\not=\emptyset}$).  We additionally require that
the list $L_s$ is topologically sorted according to inclusion. We say that
$p$-intersection queries can be answered with \emph{overhead}
$(T_{inter}(p),T_{updt}(p))$ when dynamically adding or removing a set from the
collection takes time $T_{updt}(p)$ at most and the $p$-intersection query for
any set $s\in\D$ takes time $T_{inter}(p)+|L_s|T_H$ at most.  In the case of
$d$-dimensional multi-ranges, a segment tree allows to answer $p$-intersection
queries with overhead $(O(\log^d p), O(\log^d
p))$~\cite{BergGeomBook,intersectionBoxes}.  In the case of wildcard
expressions, a trie or a Patricia tree allows to answer $p$-intersection
queries with overhead $(O(\ell p), O(\ell))$ (the whole tree has to be
traversed in the worse case, but no sorting is necessary as the result is
naturally obtained according to lexicographic order).

\subsection*{A.2 Incremental computation of atoms (Theorem~\ref{th:main})}

The following is a refinement of Theorem~\ref{th:main} with respect
to data-structures that provide time bounds on $p$-collection operations and
$p$-intersection queries. For the sake of simplicity of
asymptotic expressions, we make the very loose assumption that
$\ell=o(m)$ and $n\le m$. (We are mainly interested in the case where $m$ is
large. Note also that examples with $m<n$ would be very peculiar.)

\begin{theorem}\label{th:mainbis}
Given a space set $H$ and
a collection $\R$ of $n$ subsets of $H$, the collection $\UC(\R)$ 
of combinations that canonically 
%inclusion-wise
 represent the atoms generated by $\R$
can be incrementally computed with $O(\min(n+k\K\log m, \k m\log m, n m) m)$
elementary set operations where
$m$ denotes the number of atoms generated by $\R$, $k$ denotes the overlapping
degree of $\R$, $\k$ denotes the average overlapping
degree of $\R$ and $\K$ denotes the average overlapping degree of $\C(\R)$.

More precisely, if the data-structures used for representing sets 
and collections of sets enable elementary set operations within time $T_{set}$,
$p$-collection operations within time $T_{coll}(p)$ and $p$-intersection
queries with overhead $(T_{inter}(p),T_{updt}(m))$, then the representation
of the atoms generated by $\R$ can be computed in 
$O(n T_{inter}(m)  + \k m T_{updt}(m)  
  + \min(k\K,km) m (T_{set} + T_{coll}(m)))$ time.
\end{theorem}

%% Note that there is no equality but inclusion: some of the pre-existing
%% uncovered combinations may become covered.  This is stated by the following
%% lemma.
The following lemma formally states that uncovered combinations
of $\R\cup\set{r}$ can be obtained from $\UC(\R)$ and justifies the
overall approach of Algorithm~\ref{algo:add}.

\begin{lemma}
\label{lem:addrule}
Given a new rule $ r \subseteq H $, the collection $ \UC(\R') $ of
uncovered combinations of $ \R' = \R \cup \{r\} $ can be obtained by
intersecting uncovered combinations in $\UC(\R)$ with $ r $. More
precisely, we have $ \UC(\R') \subseteq \UC(\R) \cup \{c\cap r | c \in
\UC(\R)\} $.
\end{lemma}

\begin{proof}
Consider an uncovered combination $ c'\in\UC(\R') $.  Let $ c $ be the
intersection of the containers of $c'$ in $ \R $: $ c = \cap_{r \in \R(c')} r
$. We have either $ c' = c $ if $\R'(c') = \R(c') $ or $ c' = c\cap r $
if $ \R'(c') = \R(c') \cup r$.  To conclude, it is thus sufficient to
show that $ c $ is uncovered in $ \R$.  This follows from $c'\subseteq c$ and
$\R'(c')\subseteq\R(c)\cup\set{r}$: the non-containers of $c$ in $\R$ are
also non-containers of $c'$ in $\R'$ and if $c$ was covered, so would be $c'$.
%% As $ c' $ is characteristic in $ \R' $, we have $\paren{\cap_{r\in
%%     \R'(c')}r} \cap \paren{\cap_{r\in \R' \setminus
%%     \R'(c')}\overline{r}} \neq \emptyset$. $ \R(c) $ and $ \R
%%     \setminus \R(c) $ are included in $ \R'(c') $ and $ \R'
%%     \setminus \R'(c') $ respectively (one of the two differs by $
%%     r $), so $\paren{\cap_{r\in \R(c)}r} \cap
%%     \paren{\cap_{r\in \R \setminus \R(c)}\overline{r}} \neq
%%     \emptyset$: $ c $ is not covered in $ \R $ and belongs to $
%%     \CS(\R) $.
\end{proof}

The following Lemma, which expresses any combination as a disjoint union 
of atoms justify the computation of atom cardinalities in 
Algorithm~\ref{algo:add} by scanning $c.sup$ lists for $c\in\UC$.
\begin{lemma}\label{lem:cardinals}
Given a combination collection $\C'\subseteq \C(\R)$ containing $\UC(\R)$, we
have $d=\cup_{c\in \C' \mid c \subseteq d}a(c)$ for all $d\in \C'$.
This union is disjoint and 
we have $|a(d)|=|d| - \sum_{c\in \C' \mid c \subsetneq d}|a(c)|$.
\end{lemma}

\begin{proof}
Reminding that any combination $d$ includes 
$ a(d) = d\cap\paren{\cap_{r\in\R\setminus\R(d)}\overline{r}}$, 
we have $\cup_{c\in \C'  \mid c \subseteq d}a(c) \subset d$.  
Conversely, consider $ h\in d$. The
sets of $\R$ containing $h$ are $ \R(\set{h}) $, and we have $ h\in a(c) $ for
$c=\cap_{r\in \R(\set{h})}r $. As $ \R(d)\subseteq \R(\set{h}) $ (the sets
that contain $ d $ also contain $ h $) and $ d=\cap_{r\in \R(d)}r $, we have 
$c\subseteq d $. Hence $ d\subset\cup_{c\in \C' \mid c \subseteq d}a(c) $.
This union is disjoint as each $a(c)$ is either an atom or is empty.
\end{proof}

We can now state the following proposition about the guarantees of
Algorithm~\ref{algo:add}.

\begin{proposition}\label{prop:incr}
Algorithm~\ref{algo:add} allow to dynamically
update the collection $\UC$ of uncovered combinations of a collection $\R$
using $O(m+\sum_{a\in A_r}\card{\UC'(a)}\log m)$ elementary set operations when
a rule $r$ is added to $\R$,  where
$m$ denotes the number of atoms of $\R$, 
$A_r=\set{a\in\A(\R') \mid a\subseteq r}$ denotes the
atoms of $\R'$ included in $r$,
and $\UC'(a)$ denotes the uncovered combinations of $\R'$ that contain $a$.

More precisely, if the data-structures used for representing sets 
and collections of sets enable elementary set operations within time $T_{set}$,
$p$-collection operations within time $T_{coll}(p)$ and $p$-intersection
queries with overhead $(T_{inter}(p),T_{updt}(m))$, 
then the update of $\UC$ can be performed
in $O(T_{inter}(m) + \card{A_r} T_{updt}(m) 
  + (T_{set}+T_{coll}(m)) \sum_{a\in A_r}\card{\UC'(a)})$ time.
\end{proposition}

\begin{proof}
As discussed before the correctness of Algorithm~\ref{algo:add} for obtaining
uncovered combinations after adding set $r$ to a collection $\R$ from
$\UC=\UC(\R)$ and $\set{c\cap r\mid c\in \UC}$ results from
Lemma~\ref{lem:addrule}. For a new combination $c$, the $c.sup$ list is
obtained from $c.parent.sup$ and $c.sup$ is updated similarly for $c\in\UC$
such that $c\subseteq r$. The correctness of this approach has already been
discussed in Subsection~\ref{sec:algo}. We develop here the key argument for
ignoring a combination $c'=c\cap r$ when it is produced by several minimal
elements $c_1,\ldots,c_i\in\UC$ such that $c_j\cap r=c'$ for $j$ in $1..i$. If
this happens, we know that $\cap_{j\in 1..i}c_j$ is not in $\UC$, meaning that
it is covered in $\R$ and so is $c'\subseteq \cap_{j\in 1..i}c_j$ in
$\R\cup\set{r}$. We can thus safely eliminate $c'$ in the first
phase of the algorithm.  For the remaining new combinations $c'$, the parent
$c$ of $c'$ is the unique combination $c\in\UC$ such that $c\cap r=c'$ and
which is minimal for inclusion.

The correctness of the atom cardinality computation 
follows by induction on the number combinations in $\mb{Incl}$ processed so
far in the corresponding for loop. Consider a newly created combination $c$.
The initial value of $c.atsize$ is $|c|$.
Assuming that the correct value $b.atsize$ has been obtained for
$b$ processed before $c$ and $c\in b.sup$, $|a(b)|$ has been subtracted from
$c.atsize$ and
Lemma~\ref{lem:cardinals} implies that
$c.atsize=|a(c)|$ when we consider $c$ in the for loop. 
%he correct value will be subtracted from
%$d.atsize$ for $d\in c.sup$ such that $d$ is new or $d$ is the parent of $c$.
For $c$ already in $\UC$ before adding $r$ and for $b\subseteq c$ 
processed before $c$,
$b.atsize$ has been subtracted from
$c.atsize$ only for newly created $b$. For $b\in\UC$, $|a(b)|$ may have
decreased but this difference is compensated by 
$\sum_{b'\in\mb{New} \mid b'\subsetneq b}|a(b')|$. This is the reason why
Algoritihm~\ref{algo:add} updates only $c.parent.atsize$ besides
$d\in c.sup$ such that $d\in\mb{New}$.
The correctness of the atom cardinality computation implies that all 
covered combinations are removed and the correctness of 
Algorithm~\ref{algo:add} follows.

\medskip

We now analyze the complexity of Algorithm~\ref{algo:add}. 
The bound in terms of elementary set operations is obtained when 
balanced binary search tree (BST for short)
are used to store the various collections of sets
(i.e. $\UC$, $\mb{Incl}$, $\mb{New}$ and $c.sup$ for $c\in \UC$).
%% We first assume that balanced binary search trees are
%% used to store sets in $\UC$, $\mb{Incl}$ and $\mb{New}$ allowing
%% to add/remove a set or test membership within a logarithmic number of 
%% elementary set operations. We suppose that a sorted list is used to
%% store the combinations in $c.sup$ that strictly contain $c$.
When adding a set $r$,
finding the combinations in $\UC$ that intersect $r$ is a $m$-intersection
query and can be performed in $O(T_{inter}(m)+|\mb{Inter}|T_{set})$ time
or $O(m\log m)$ set operations using BST (sorting is only necessary in that
case).
%and computing $\set{(c\cap r,c)\mid c\cap r\not=\empty}$
%requires $O(m)$ elementary set operations. 
%The result $\mb{Inter}$ can be sorted with $O(m\log m)$
%with a binary search tree constructed with the appropriate order.
The collection $\mb{Incl}$ is then constructed in 
$O(|\mb{Inter}|\log m)$ operations with BST
or $O(|\mb{Inter}|(T_{set}+T_{coll}(m)))$ with appropriate data-structures.
%can be done in logarithmic number of operations using a binary search tree
%for storing $\C$ and $\mb{Incl}$.
%using a simple sorted list structure for $c'.sup$
%(again using ascending order of cardinality). 
Removing combinations $c$
such that $c.covered=true$ is just a matter of scanning $\mb{Incl}$
again and can be done within the same complexity.
% using a binary
%search tree for storing the collections $\C$ and $\mb{Incl}$.
Let $\I$ denote the combinations included in $\mb{Incl}$ at that point
(just before cardinality computations).
%Let $d'=|\I|$ denote the size of $\I$.
The computation of $c.sup$ for $c\in\I$ is done only when $c.atsize > 0$, i.e.
only if $c$ represents one of the atoms in $A_r$. 
we thus have $\card{I}\le \card{A_r}$.
%included in $r$. 
This requires at most $O(\card{c.parent.sup})$ operations.
Note that for each uncovered combination $d\in c.parent.sup$ yields at
least one uncovered combination in $c.sup$ ($d$ itself or $d\cap r$ or both).
We thus have $\card{c.parent.sup}\le\card{\UC'(a(c))}$.
%can then be done
%in $O(d'\K\log d')$ by merging $c.sup$  with $\set{d\cap r\mid d\in c.sup\mbox{ and }d\cap r\in\mb{Incl}}$ for each $c\in\mb{Incl}$ 
%(the $\log d'$ factor comes from testing membership in $\mb{Incl}$).
The overall computation of $sup$ lists can thus be performed within 
$O(\sum_{a\in A_r}\card{\UC'(a)}\log m)$ set operations with BST
and $O((T_{set}+T_{coll}(m))\sum_{a\in A_r}\card{\UC'(a)})$ 
time with appropriate data-structures.
The computation of class cardinalities and the removal
of covered combinations from the $sup$ lists have same
complexity. 
Removal of covered combinations from $\UC$ and
$\mb{Incl}$ takes $O(\card{A_r})$ collection operations and fits within the same complexity
bound. Additional cost of $\card{A_r} T_{updt}(m)$ is necessary when
maintaining data-structures enabling efficient $m$-intersection queries.
The whole algorithm can thus be performed
in $O(T_{inter}(m) + \card{A_r} T_{updt}(m) + |\mb{Inter}| (T_{set}+T_{coll}(m))
  + (T_{set}+T_{coll}(m)) \sum_{a\in A_r}\card{\UC'(a)})$ time
or using $O(m + |\mb{Inter}|\log m + \sum_{a\in A_r}\card{\UC'(a)}\log m)$ 
elementary set operations with BST.

To achieve the proof of complexity of Algorithm~\ref{algo:add}, we show
$|\mb{Inter}|\le \sum_{a\in A_r}\card{\UC'(a)}$. 
Consider a combination $c\in\UC$ that intersects
$r$. Then $c$ can be associated to an  atom $c.atm$ of $A_r$
included in $c\cap r$ (such atoms exist according to
Lemma~\ref{lem:cardinals}). For any atom $a\in A_r$, let $a.par$ 
denote the atom in $\A(\R)$ that 
contains $a$ (we have $a=a.par$ or $a=a.par\cap r$).
Now for  $c\in\mb{Inter}$, consider the atom $a=c.atm\in A_r$.
As $a.par\in\A(\R)$ is an atom and $c\in\C(\R)$ is a combination intersecting
$a.par$, we have 
$a.par\subseteq c$ and $c\in\UC(a.par)$ is one of the combinations 
in $\UC(\R)$ that
contains $a.par$. For each such combination $c$, $a(c)$ intersects $r$ or
$\overline{r}$ (or both), and $c$ or $c\cap r$ is uncovered in 
$\R'=\R\cup\set{r}$. Both contain $a$ and we have 
$\card{\UC(a.par)}\le \card{\UC'(a)}$. We can thus write
$|\mb{Inter}|=\sum_{a\in A_r}\card{\set{c\in\mb{Inter} \mid c.atm = a}}
\le\sum_{a\in A_r}\card{\UC(a.par)}\le\sum_{a\in A_r}\card{\UC'(a)}$.
%% Each one of theses atoms is included in $r$ too so their number is at most
%% $d_r$. Consider one of these atoms $a$ associated to such a combination
%% $d$. Since atoms are disjoint, there is a unique atom $a^{\prime}$ of $\R$
%% possibly equal to $a$ such that $a=a^{\prime} \cap r$.  Note that because
%% $a^{\prime}$ is an atom of $\R$ that intersects $d$, we must have
%% $a^{\prime} \subseteq d$. To sum up, there are at most $d_r$ such atoms
%% $a^{\prime}$, each one being included in at most $\K$ combinations $d$ of
%% $\R$, we have $|\mb{Inter}|\le \K d_r$.
\end{proof}

It should be noted that in the case of $(d,\ell)$-multi-ranges, we can 
distinguish in the analysis, the computation of $|c|$ for new $c$ that
costs $O(\ell\log \ell)$ from other cardinality manipulations (subtractions and
comparisons) that take $O(\ell)$. This allows to obtain the bound claimed
in Theorem~\ref{th:main}.

Theorem~\ref{th:mainbis} results from iteratively applying 
Algorithm~\ref{algo:add} (or Algorithm~\ref{algo:basic})
for each set of $\R=\set{r_1,\ldots,r_n}$.

\begin{proof}[Proof of Theorem~\ref{th:mainbis}]
From Proposition~\ref{prop:incr}, the overall complexity of atom computation is\\
$O(\sum_{i=1..n}\paren{m_i+\sum_{a\in A_i}\card{\UC_i(a)}}\log m)$
set operations where $m_i$
denotes the number of atoms in $\A(\set{r_1,\ldots,r_{i-1}})$ 
and $A_i$ denotes the
atoms of $\A(\set{r_1,\ldots,r_i})$ included in $r_i$ and
$\UC_i=\UC(\set{r_1,\ldots,r_i})$.  We first consider the case where $\K$ is
unbounded (it is possible to construct examples with $m=n$ and
$\K=\Omega(2^n)$).  As we add a set to $\R$, the number of atoms can only
increase (each atom remains unchanged or is eventually split into two). We thus
have $m_i\le m$ and $\card{A_i}\le\card{\set{a\in\A(\R) \mid a\subseteq r_i}}$.
Using $\card{\UC_i(a)}\le m_{i+1}\le m$, the overall complexity is 
$O(nm+m\log m\sum_i\sum_{a\in\A(r)}\card{\set{r\in\R\mid a\subseteq r}})
=O(nm+\k m^2\log m)$ by definition of average overlap.
The $O(nm^2)$ bound is obtained by using Algortihm~\ref{algo:basic}
instead of Algorithm~\ref{algo:add}.

We now derive a bound depending on the average overlapping degree $\K$
of combinations.
Consider an atom $a\in A_i$ and an uncovered combination $c\in\UC_i(a)$. We can
associate $a$ to an atom $a.desc\subseteq a$ in $\A(\R)$. As $c$ is also a
combination in $\C(\R)$, we have $c\in \C(a.desc)$ where $\C(s)$ denotes the
combinations of $\R$ containing $s$. As the atoms in $A_i$ are disjoint, the
atoms $a.desc$ for $a\in A_R$ are pairwise distinct.  We thus have $\sum_{a\in
  A_i}\card{\UC_i(a)}\le \sum_{a\in \A(\R)\mid a\subseteq r_i} \card{\C(a)}$.
The overall complexity of atom computation is $O(nm+\log m\sum_{a\in
  \A(\R)}\sum_{c\in \C(a)}\card{\set{i\mid a\subseteq r_i}})=O(nm+k\log
m\sum_{a\in\A(\R)}\card{\C(a)})=O(nm+k\K m\log m)$ by definition of overlapping
degree and average overlapping degree respectively.  The refined bound in terms of
$T_{set},T_{coll}(m),T_{inter}(m),T_{updt}(m)$ is obtained similarly.
%% As the number of atoms included in set $r_i$ can only increase
%% as we add more sets, each $d_i$ is bounded by the number $d_{r_i}^-$ 
%% of atoms of $\R$ included in $r_i$.
%% We can then use the fact that the number of inclusion relations
%% between uncovered combinations is $\sum_{c\in\UC(\R)}d_c^-=\sum_{c\UC(\R)}d_c^+$
%% where $d_c^-=\card{\set{b\in\UC(\R)\mid b\subsetneq c}}$
%% and $d_c^+=\card{\set{d\in\UC(\R)\mid c\subsetneq d}}$.
%% As we have $\K =\max_{c\in\UC(\R)}d_c^+$, the get the
%% bound $\sum_{i=1..n}d_{r_i}^-\le m\K $. The overall complexity is
%% thus $\O(nm+m\K^2)$ elementary set operations.
\end{proof}

\subsection*{A.3 Application to forwarding loop detection (Corollary~\ref{cor:loop})}

Corollary~\ref{cor:loop} directly follows from the following claim and
Theorem~\ref{th:main}.

\begin{claim}
Given the collection $\R$ of rule sets of a network $\mathcal{N}$, and for each
atom $a\in\A(R)$ the list $\R(a)$ of sets in $\R$ that contain $a$,
forwarding loop detection can be solved in $O(\k n_Gm)$ time where
$m=\card{\A(\R)}$ is the number of header classes, $\k$ is the average
overlapping degree of $\R$ and $n_G$ is the number of nodes in $\mathcal{N}$.
\end{claim}

\begin{proof}
For that, we assume that each rule set $r\in\R$ is associated with the list
$L_r$ of forwarding rules $(r,a)$ that have rule set $r$.  Each such rule is
also supposed to be associated to the node $u$ whose table contains it and the
index $i$ of the rule in $T(u)$. Each list $L_s$ is additionally supposed to be
sorted according to associated nodes.  Such lists can easily be obtained by
sorting the collection of all forwarding tables according to the predicate
filters of rules.
%% (Sorting forwarding tables takes 
%% $O(nn_G\log\min(n,n_G))$ elementary set operations in general and
%% $O(\ell nn_G)$ time in the case of wildcard of multi-range rules by using
%% a trie structure.) 

The claim comes from the fact that uncovered
combination in $\UC(\R)$ %inclusion-wise 
represent atoms of $\A(\R)$. 
It follows from testing for each header
class $a\in\A(\R)$ whether the graph $G_a=G_h$ for all $h\in a$ 
has a directed cycle.
$G_a$ is computed by merging the lists $L_s$ for $s\in\R(a)$ in time
$O(\card{\R(a)}n_G)$. This graph has at most $n_G$ edges and cycle detection
can be performed in $O(n_G)$ time. The overall complexity follows from
$\k m=\sum_{a\in\A(\R)}\card{\R(a)}$ by definition of $\k$.
\end{proof}

\section*{Appendix B : Difficult inputs for previous works}

\subsection*{B.1 Veriflow}

Verfilow~\cite{veriflow} represents the complementary of a multi-range
$r=[a_1,b_1]\times\cdots\times[a_d,b_d]$ as the union of 
$2d-1$ multi-ranges (at most):
\begin{itemize}
\item $[0,a_1-1]\times H_{2..d}$ and $[b_1+1,\infty_1]\times H_{2..d}$,
\item $[a_1,b_1]\times [0,a_2-1]\times H_{3..d}$ and
      $[a_1,b_1]\times [b_2+1,\infty_2]\times H_{3..d}$,
\item $\cdots$,
\item $[a_1,b_1]\times \cdots\times [a_{d-1},b_{d-1}]\times [0,a_d-1]$ and
      $[a_1,b_1]\times \cdots\times [a_{d-1},b_{d-1}]\times [b_d+1,\infty_d]$,
\end{itemize}
where $\infty_i$ denotes the maximum
possible value in field $i$,
and $H_{i..j}=[0,\infty_i]\times\cdots\times [0,\infty_j]$ denotes the
multi-range of all possible values for fields $i,\ldots,j$ for
$1\le i \le j \le d$.

The difficult input for Veriflow 
%(as indicated in Subsection~\ref{sec:exp-verif}) 
consists in a network
with $n=dp+1$ rules associated to the following multi-ranges:
\begin{itemize}
\item $r_0=H_{1..d}$,
\item $r_i^j=H_{1..i-1} \times [a_j,a_j] \times [b,b]^{d-i}$ 
for $i,j\in [1..d]\times [1..p]$.
\end{itemize}
%where $C_{i..j}$ denotes the singleton multi-range $[\infty_i,\infty_i]\times\cdots\times [\infty_j,\infty_j]$ for
%$1\le i \le j \le d$. 
%The sets $r_i^j$ for $i,j\in [1..d]\times [1..p]$ are pairwise disjoint. The
%number of header classes is thus $m=dp+1=n$.

Consider the sub-classes generated while computing $r_0\cap \paren{\cap_{i,j\in
    [1..d]\times [1..p]} \overline{r_i^j}}$.  The union of multi-ranges
representing $\overline{r_i^j}$ contains in particular $H_{1..i-1}\times
[0,a_j-1]\times H_{i+1..d}$ and $H_{1..i-1}\times [a_j+1,\infty_i]\times
H_{i+1..d}$.
% where $A_i=[0,a_j-1]$ and $B_i=[a_i+1,\infty_i]$. 
This implies that Veriflow generates on such an input
all $p^d$ sub-classes of the form $I_1\times\cdots\times I_d$ 
with $I_i=[a_j+1,a_{j+1}-1]$ for some $j\in [0,d]$ (we set $a_0=-1$).
Forwarding loop detection of an $n_G$-node network thus requires
$\Omega(p^d n_G)=\Omega(\paren{\frac{n}{d}}^dn_G\frac{m}{d})$ time for Veriflow.
%% If all fields have at least $\log
%% \frac{2n}{d}$ bits, it possible to repeat $\frac{n}{d}$ times this construction
%% (with increasing values for all $a_i,b_i$) to obtain $n$ pairwise disjoint
%% rules (excluding $r_0$) that generate $m=n+1$ header classes but where Veriflow
%% produces $\frac{n}{d}2^d$ sub-classes to test. Each sub-class test consists in
%% constructing the forwarding graph for headers in the sub-class 
%% (in $\Theta(\average{k}n_g)$ time) and search it for a cycle. The overall
%% loop detection of Veriflow thus requires 
%% $\Omega(\frac{2^d}{d}\average{k}n_G m)$ time on such examples.
As this example has overlappingd degree 2, this justifies the two lower-bounds
indicated for Veriflow in Table~\ref{tab:comp} for $d$-multi-ranges.

It is possible to adapt Veriflow to support general wildcard matching by
considering each field bit as a field. The
wildcard expressions 
$r_0=*^\ell,r_1=01^{\ell-1},\ldots, r_\ell=0^{\ell-1}1$ will then 
similarly generate all $2^{\ell/2}$ sub-classes obtained by concatenation of
words $10$ and $11$. This justifies the two lower-bounds
indicated for Veriflow in Table~\ref{tab:comp} for $\ell$-wildcards.

\subsection*{Appendix~B.2 : HSA/NetPlumber approach}

The NetPlumber approach could be generalized to more general types of rules.
However, we show that the simple heuristic for emptiness tests described
in \Cref{sec:netpl-exp} is not sufficient.
%Subsection~\ref{sec:linfrag}, bounded overlapping degree would then guarantee
%polynomial time execution. Interestingly, we show that this is not the case
%when overlapping degree is not bounded.  
We provide an example where the HSA/NetPlumber approach generates an
exponential number of paths while the number of classes is linear if it relies
solely on this heuristic.  Consider header space $H=\set{1..n}$ and the
following $n+1$ rule sets:
$r_1=\overline{\set{1}},\ldots,r_n=\overline{\set{n}}$ and $r_{n+1}=H$.
Consider a network $\mathcal{N}$ with $n_G=n(n+1)$ nodes. Each node $u_{i,j}$
for $0\le i\le n$ and $1\le j\le n$ has table
$T(u_{i,j})=(r_{i+1},FwdD_{i+1,1}),\ldots,(r_n,FwdD_{i+1,n}),(r_{n+1},FwdD_{i+1,n})$
where action $FwdD_{i,j}$ indicates to forward packets to node $u_{i,j}$ for
$i\le n$ and to drop packets for $i=n+1$.  Starting from $u_{0,1}$, the HSA
approach generates a path for each combination $r_{i_1}\cap\cdots\cap r_{i_p}$
for $p\le n$ and $1\le i_1<\cdots<i_p\le n$. This path goes through
$u_{0,1},u_{1,i_1},\ldots,u_{p,i_p}$ and then through
$u_{p+1,n},\ldots,u_{n,n}$. It is constructed at least for term
$r_{i_1}\cap\cdots\cap r_{i_p}\setminus \cup_{j\notin\set{i_1,\ldots,i_p}}
r_j$. The heuristical emptiness test of NetPlumber does not detect that it is
empty since $r_{i_1}\cap\cdots\cap r_{i_p}$ contains $j$ for $j\notin
\set{i_1,\ldots,i_p}$ and it is not included in $r_j$.  The number of paths
generated is thus at least $\sum_{1\le p\le n}{n \choose p}=2^n-1$. However,
the header classes are all singletons of $H$ and their number is $m=n$. Note
the high overlapping degree $k=n$ of this collection of rule sets.

\subsection*{Appendix~C : Linear fragmentation versus overlapping degree}
%\vspace{-.3cm}

Interestingly, a complexity analysis of HSA loop detection is given in the
technical report~\cite{hsa-report} under an assumption called ``linear
fragmentation''. This assumption, which is based on empirical observations,
basically states that there exists a constant $c$ such that a given 
routing path $P$ will branch at most $c$ times. More precisely it states that
a term $s=c_0\setminus \cup_{i=1..p} c_i$ in the expression representing
$H_P$, the set of headers that can follow $P$, 
intersects at most $c$ of the rule sets in the table of the end node of $P$.
A simple induction allows to bound
the number of terms generated by all paths of $p$ hops from a given source by 
$c^pn$.  Under linear-fragmentation, the time complexity of HSA loop
detection (excluding emptiness tests) is thus proved to be $O(c^{D_G}D_Gn^2m_G)$ 
in \cite{hsa-report} where $D_G$ is the diameter
of network graph $G$, $n$ the number of rules, and $m_G$ the number of ports in
$G$ (in our simplified model each node has a single input port and $m_G=n_G$
the number of nodes in $G$). 
It is then argued that in practice the constant
$c$ gets smaller as the length $p$ of the path considered increases and that
practical loop detection has complexity $O(D_Gn^2m_G)$ as claimed
in~\cite{hsa}. However, it is not rigorous to neglect the (exponential)
$c^{D_G}$ factor under the sole linear-fragmentation hypothesis.
%% Moreover, the above complexity bound accounts only for 
%% loop detection of formal expressions of header classes.
%% It does take into account the emptiness test required for each such expression
%% encountered. We show in Section~\ref{sec:comparison} 
%% that a single emptiness test can be exponential
%% in $\ell$ because of complement computations.

%Concerning linear fragmentation,
Additionally,
we think that low overlapping degree provides a simple explanation for the
phenomenon observed by Kazemian et al.: as the path length increases, the
terms representing the header that can traverse the path
result from the intersection of more
rules and become less likely to intersect other rules when overlapping degree
is limited. Moreover, bounded overlapping degree $k$ implies that the
number of terms generated by HSA within $p$ hops is bounded by 
$O(n^{\min(p,k)})$. The total number of terms generated is thus
bounded by $O(n_G n^{k})$. This guarantees
that all HSA computations besides emptiness tests 
remain polynomial for constant $k$. In contrast, with
the example provided in Appendix~B.2 (which has unbounded overlapping degree),
the HSA approach can
generate exponentially many paths compared to the number of header classes
in the context of general rules.

\subsection*{Appendix~D : Related notion of weak completeness}

In the context of resolution of conflicts between rules, 
Boutier and Chroboczek~\cite{boutier} introduce the concept of
\emph{weak completeness}: a collection $\R$ is weakly complete
iff for any sets $r,r'\in \R$, we have
$r\cap r' = \cup_{r''\subseteq r\cap r'} r''$. They show that this is a
minimal necessary and sufficient condition 
for all rule conflicts to be solved when priority of
rules extends inclusion (i.e. $r$ has priority over $r'$ when 
$r\subsetneq r'$).
Interestingly, we can make the following connection with this work:
%One can easily see that 
%Lemma~\ref{lem:cardinals} implies
%the following characterization of $a(c)$ for any combination $c$:
%\begin{corollary}\label{cor:cover}
given a combination collection $\C'\subseteq \C(\R)$ containing $\UC(\R)$, we
have $a(c)=c\setminus\cup_{c'\in \C' \mid c' \subsetneq c}c'$ for all $c\in\C'$.
(See Lemma~\ref{lem:cardinals} in Appendix~A.2.)
%\end{corollary}
%
%% \begin{proof}
%% According to Lemma \ref{lem:cardinals}, we have $a(c)=c\setminus\cup_{c'
%%   \subsetneq c}a(c')$ (recall that classes are disjoint). As the same lemma
%% also implies $c'=\cup_{c''\subseteq c'}a(c'')$, we have $\cup_{c' \subsetneq
%%   c}c' = \cup_{c'' \subsetneq c}a(c'')$, the claim follows.
%% \end{proof}
%
%The above corollary 
%This allows to make the following connection with the notion of weak 
%completeness:
% introduced by Boutier and Chroboczek~\cite{boutier}.
%According to their terminology, 
This allows to show that $\UC(\R)$ is weakly complete.
It is indeed the smallest collection of combinations of $\R$ that contains 
$\R\cup\set{H}$ and that is weakly complete. Our work thus also provides 
an algorithm for computing such an optimal ``weak completion''.
%% That is for any $c,c'\in\UC(\R)$, 
%% $c\cap c'\subseteq \cup_{b\in \UC(\R) \mid b \subseteq c\cap c'}b$.
%% Indeed, if $c\cap c'$ is uncovered, it is in $\UC(\R)$,
%% otherwise $a(c\cap c')=\emptyset$ and we have the desired property
%% according to the corollary.

\subsection*{Appendix~E : Topological header classes}

Inspired by the HSA/NetPlumber~\cite{hsa,netplumber} approach, we can 
refine the definition of header classes. We fix a source node $s$. A
forwarding path $P$ originating from $s$ is a sequence 
$(u_0,i_0),\ldots,(u_p,i_p)$ where $u_0,\ldots,u_p$ are the nodes
encountered along the path
and $i_0,\ldots,i_p$ are the indexes of the
forwarding rules followed: rule $i_0$ is applied at node $u_0=s$, 
then rule $i_1$ at node $u_1$ and so on. We then define two headers as 
topologically equivalent from $s$ if they follow the same forwarding paths.
Note that each topological class for that relation corresponds to a unique
path (if some headers match two distinct rules at a node, 
the first one is followed).

We now show that the topological classes (from $s$) are
indeed certain atoms of partial collections of rule sets.
Given a forwarding path $P=(u_0,i_0),\ldots,(u_p,i_p)$.
Let $\R^j=\{r^j_1,\ldots,r^j_{n_j}\}$ denote the collection of rule sets
of node $u_j$ where $r^j_1,\ldots,r^j_{n_j}$ corresponds to the order of the
$n_j$ rules in $T(u_j)$. Let also $\R^j_i=\{r^j_1,\ldots,r^j_i\}$
denote the partial collection of the $i$ first rule sets.
The set of headers that can follow rule $i$ in node $u_j$
is $s^j_i=\overline{r^j_1}\cap \cdots \cap \overline{r^j_{i-1}} \cap r^j_i$. 
If $P$ corresponds to a topological class $cl(P)$, we have
$cl(P)=s^0_{i_0}\cap\cdots\cap s^p_{i_p}$.
As the only positive terms of this intersection are $r_{i_0},\ldots,r_{i_p}$,
$cl(P)$ is indeed the atom $a(c_P)$ associated to the combination
$c_P=r_{i_0},\ldots,r_{i_p}$ in $\A(\R^0_{i_0}\cup\cdots\cup \R^p_{i_p})$.
Note that we consider a different partial collection
$\R^0_{i_0}\cup\cdots\cup \R^p_{i_p}$ for each path.
Our framework can thus be applied to test whether the combination
$c_P$ associated to any path $P$ does correspond to a topological class
by testing the emptiness of $a(c_P)$. We detect a forwarding loop (starting
from $s$) as soon as a path $P$ reaches a former node of $P$.

Our incremental algorithm is well suited for incrementally 
augmenting the collection of rules as we progress along the path.
Using a persistent style implementation for the data structure 
storing collections (as classically done when implementing binary search
trees in functional languages), we can step back at a branching node $v$
and follow a different branch without having to recompute the collection
of atoms generated by the path up to $v$. The collection is also
incrementally augmenting as we progress in routing table of a branching
node $v$ and explore different paths. This allows to efficiently
perform a search of the graph in a depth first search manner.

It is important to note that two different explored paths $P$ and $Q$ 
are associated to different combinations $c_P\not= c_Q$.
To see this, suppose that $P$ and $Q$ branch at node $u$:
$P$ follows rule $r_i$ of $T(u)$ while $Q$ follows $r_j$ with $i<j$.
After $(u,i)$,  path $Q$ cannot follow any rule associated 
with set $r_i$ due to emptiness tests. This ensures that the number
of paths generated is bounded. Similarly to the proof of 
Proposition~\ref{prop:incr} and Theorem~\ref{th:mainbis}, 
each combination $c_P$ generated by a path $P$
can be associated to an atom of $\A(\R)$ where
$\R$ denotes the collection of all rule sets in the network.
The number of topological classes is thus always bounded by $m=|\A(\R)|$.
However, for a given collection $\R$ of rule sets generating $m$ atoms, 
one can easily produce a network topology with $m$ topological classes.

When performing the search from source $s$,
the number of paths (and prefixes of paths) generated is
at most $\sum_{a\in \A(\R)}|\UC(a)| = K m \le nm$ where 
$\UC=\UC(\A)$ are the uncovered combinations generated by $\R$. 
As each path prefix accounts for one call of
Algorithm~\ref{algo:basic}, this search thus costs
$O(n m^3)$ elementary set operations.
Repeating this for all possible source nodes, we get an overall complexity
of $O(n_G  n m^3)$ operations for forwarding loop detection. 
(This complexity analysis can be refined
in terms of overlapping degree and optimized by using potentially 
Algorithm~\ref{algo:add} instead of Algorithm~\ref{algo:basic}.)

Note that the number of paths explored could grow exponentially 
without appropriate emptiness tests as exemplified in Appendix~B.2.
Although this approach remains polynomial in terms of $n$ and $m$, 
its complexity guaranties are less interesting than what we propose in
Section~\ref{sec:algo}. However, it allows to extend our framework with
write actions as detailed in the next appendix.

\subsection*{Appendix~F : Write actions}

First note that allowing any write operations and variable header length
(by allowing to push sub-headers) make the forwarding loop detection problem
undecidable as we can then easily simulate a pushdown automaton with several
stacks. However, in the context of MPLS, it is natural to allow to push and pop
MPLS headers (and only that type). In the classical functioning of MPLS, 
each push or swap action always writes a fixed label value, 
rules taking into account the MPLS header always have
higher priority and base their decisions on the outermost label only
(and no other field).
In that case, paths with MPLS forwarding can be factored out by adding 
shortcut edges (from the first push action to last pop action) in a preliminary
step which should also include loop detection for each label value pushed
somewhere in the network. This can clearly be performed in 
polynomial time. We reserve the study of more general write action with push and
pop actions for future work. However, we now sketch how to handle
write actions in the context of fixed length headers.

We suppose that a forwarding rule applied to a packet 
can perform a write action before forwarding the packet. 
We typically think of write actions that write some fixed values at fixed
positions. We then see a write action
as a projection of the header space in a smaller sub-space. More generally,
we assume that each write action $w$ is associated with a function 
$p_w : H \rightarrow H$ with the following properties 
(we set $p_w(s)=\set{p_w(h) \mid h\in s}$ for $s\subseteq H$) :
\begin{itemize}
%% \item inclusion preservation:
%%   for $s\subseteq s'$, we have $p_w(s) \subseteq p_w(s')$;
\item intersection preservation:
  for $s\cap s'\not=\emptyset$, we have $p_w(s\cap s') = p_w(s) \cap p_w(s')$;
\item two consecutive write operations $w$ and $w'$ are equivalent to
  a single write operation $w''$ with $p_{w''} = p_{w'} \circ p_w$,
  $w''$ is then called a write pattern assumed to be efficiently computable.
\end{itemize}
We do not need to formally assume that $p_w$ is a projection (i.e. $p_w\circ
p_w=p_w$) although this is typically the case.

We additionally make the following assumption with respect the data-structure
$\D$ used for sets:
\begin{itemize}
\item $p_w(s)$ is in $\D$ and can be computed in $O(T_\ell)$ time.
\end{itemize}
All the above requirements are clearly met for any write action 
(of fixed bits at fixed positions) in the context 
of wildcard expressions. In the context of
prefix matching we can only allow to write a prefix of bits to ensure
that a set represented by a prefix is always mapped to a set that can 
be represented by a prefix. 
We can generalize this to the context of interval
matching when low order bits are ordered first (little-endian order).
We still represent a set $s$ with an interval $[a,b]$ but we additionally
consider a write pattern $w$ and the number $i$ of bits of $w$ 
written so far.
%(the $i$ low order bits of $a$ and $b$ are then identical).
The write pattern $w$ is simply
an integer with same bit representation as $a$ and $b$.
To test whether a header $h$ is in $s$, we first check that its $i$ lowest bits
are identical to those of $w$, then shift $a$, $b$, and $h$ by
$i$ bits to erase low order bits, obtaining $a',b',h'$ respectively
and test whether $a' \le h' \le b'$. 
Other set operations are handled similarly.
Note that the computation of $p_{w'}([a,b])$ just requires to update
the write pattern which can be globally shared for a collection of sets.
A similar trick can be used with wildcard expressions.
In the context of multi-ranges,
our model thus enables prefix writes on one or several fields.
(This includes in particular writing of exact values in one or several fields.)

Forwarding loop detection can then be performed using the approach detailed
in Appendix~E by additionally decorating each node of an explored path with a
write pattern $w$ equivalent to the sequence of write actions performed so far
along the path. A forwarding loop is detected when a former node $u$ associated
to write pattern $w$ is reached again with same write pattern $w'=w$ while
growing a path $P$ containing $u$.  If the number the of write patterns that
can be generated by combining various write actions is bounded by $p$, then the
length of each path is at most $n_G p$.

Our algorithms for incremental atom computation can easily be extended in that
context: when a write action $w$ is performed when growing a path $P$, we
update the collection of uncovered combinations representing the atoms of the
current collection $\R_P$ of rule sets encountered along $P$ as follows.
Each combination $c$ is replaced by $p_w(c)$ (which typically 
amounts to a single update a the global write pattern 
in the contexts of wildcard expressions and multi-ranges).
We then recompute the inclusion relations between them and
atom sizes in $O(m^2)$ elementary set operations as in
Algorithm~\ref{algo:basic}. Combinations that are detected as
covered are removed (or just saved apart for efficient backtracking). 
We finally obtain a valid representation
of $\set{p_w(r) \mid r\in\R_P}$ since $p_w$ preserves
intersection. 

Note that write operations can only reduce the number of atoms.
Forwarding loop detection can thus be performed within a factor
$n_G p$ compared to the search procedure of Appendix~E, that is in
$O(n_G^2 n m^3 p)$ elementary set operations. This is again polynomial
in terms of number of rules $n$, number of atoms $m$ generated by the 
collection of rule sets, and the number $p$ of write patterns that can be 
generated by write operations.

\end{document}